\documentclass[11pt]{article}
\usepackage[a4paper, left=1in, right=1in,top=1in, bottom=1in]{geometry}
\usepackage[utf8]{inputenc}
\usepackage{authblk}

\usepackage[numbers,compress]{natbib}
\usepackage{graphicx}
\usepackage{amsmath,amssymb,amsfonts,amsthm}
\usepackage{physics}
\usepackage{mathtools}
\usepackage{stmaryrd}
\usepackage{adjustbox}
\usepackage{quantikz}
\usepackage[dvipsnames, table]{xcolor}
\usepackage[hypertexnames=false]{hyperref}
\hypersetup{colorlinks=true, linkcolor=blue, citecolor=blue, urlcolor=blue}
\usepackage{doi}
\usepackage{titlesec}
\usepackage{tocloft}
\usepackage{multirow}
\usepackage{makecell}

\newtheorem{definition}{Definition}
\newtheorem{theorem}{Theorem}
\newtheorem{proposition}{Proposition}

\let\abs\relax  
\DeclarePairedDelimiter{\abs}{\vert}{\rvert}
\DeclarePairedDelimiter{\ceil}{\lceil}{\rceil}
\DeclarePairedDelimiter{\iverson}{\llbracket}{\rrbracket}
\newcommand{\iv}[1]{\, \iverson*{\,#1 \,} \,}

\providecommand\given{}
\newcommand\SetSymbol[1][]{%
\nonscript\:#1\vert
\allowbreak
\nonscript\:
\mathopen{}}
\DeclarePairedDelimiterX\Set[1]\{\}{%
\renewcommand\given{\SetSymbol[\delimsize]}
#1
}

\begin{document}

\title{Unconditional and exponentially large violation of classicality}

\author[1]{Marcello Benedetti\thanks{Corresponding author, \href{mailto:marcello.benedetti@quantinuum.com}{marcello.benedetti@quantinuum.com}}}

\author[1]{Gabriel Marin-Sanchez}

\author[1]{Jordi Weggemans}

\author[1]{\authorcr Matthias Rosenkranz} 

\author[1,2,3]{Harry Buhrman\thanks{Corresponding author, \href{mailto:harry.buhrman@quantinuum.com}{harry.buhrman@quantinuum.com}}}

\affil[1]{Quantinuum, London, United Kingdom}
\affil[2]{University of Amsterdam, Amsterdam, The Netherlands}
\affil[3]{QuSoft, Amsterdam, The Netherlands}

\date{October 8, 2026}

\maketitle

\vspace{-.5cm}

\begin{abstract}
Testing the predictions of quantum mechanics has been one of the main experimental endeavors for decades. Recent advancements in technology led to a number of demonstrations which test non-classicality via specific computational tasks. Limitations of these experiments include dependence on complexity theory assumptions, susceptibility to hardware noise and inefficient verification, raising questions about their scalability. We propose to test non-classicality using a game based on complement sampling, an efficiently verifiable problem that achieves the largest possible separation between quantum and classical computation when both input and output represent samples from probability distributions. When restricting the input to instances inspired by the Bernstein-Vazirani problem, our game admits an exponentially large violation of classicality unconditionally, i.e. without relying on computational hardness assumptions. We demonstrate the game on Quantinuum System Model H2 trapped-ion quantum computers, with experiments consisting of thousands of different circuits on up to 55 qubits. The observed scores can be explained by a systematic adoption of a quantum strategy, further corroborating the quantum nature of the hardware in an efficient and scalable way.
\end{abstract}

\section*{Introduction} 

A prominent approach to validate quantum mechanics' predictions is by examining violations of Bell inequalities~\cite{Bell_1964}. Such experiments~\cite{Freedman_1972, Aspect_1982, Pan_1998, Storz_2023} can be formulated as collaborative games with two players~\cite{Brunner_2014}, often called non-local games. At each round, the referee sends a question to each player. The players can only see their own question, to which they provide an answer. The referee evaluates a function of both answers and questions to decide whether the round has been won. Depending on the resources and strategies allowed (e.g. classical, quantum, and non-signaling) the output statistics attain different maximal values.  

Tests on quantum mechanics can also be designed around computational tasks~\cite{Arute_2019, Zhong_2020, Madsen_2022, DeCross_2025, Kretschmer_2025}, where complexity theory provides the foundation to build upon~\cite{Harrow_2017}. This approach has gained traction thanks to the recent technological progress towards some of the challenges of building quantum computers. However, since error-corrected large-scale quantum computers remain inaccessible at the moment, researchers have focused on random circuit sampling (RCS), a task designed for noisy hardware. Under noise models aligned with existing quantum technologies, numerical simulations suggest that achieving a polynomial quantum speedup in RCS is restricted, at best, to circuits of few hundred qubits and depths up to one hundred~\cite{Zlokapa_2023}. Beyond 50 qubits, even if the output of the circuit still carries a signal despite the noise, it becomes very hard to verify the results due to the classical computational hardness of this. RCS suffers from either exponential classical-time verification~\cite{StilckFranca_2022}, as it requires noiseless classical simulation of the circuits, or exponential sample complexity~\cite{Hangleiter_2019}, due to the output distributions having high min-entropy. Moreover RCS relies on an unproven complexity-theoretic assumption called the quantum threshold~\cite{Aaronson_2017}.

In this work, we formulate and analyze a single-player game, which we call the \emph{complement sampling game}, significantly extending the earlier findings reported in Ref.~\cite{benedetti2025complement}. Our formulation arguably has all the desiderata for a test of quantum mechanics: i) it is unconditional in the sense that it does not rely on unproven complexity-theoretic assumptions, ii) it is efficiently verifiable by classical means, iii) it has low hardware requirement, and iv) there exists a quantum strategy that is provably exponentially better than the optimal classical strategy. 

It is natural to ask whether the complement sampling game is related to the aforementioned Bell experiments and non-local games. These typically involve multiple players measuring parts of a system that have interacted beforehand to become entangled. Our game is different because it involves a single player who has access to the full system. So while non-local games test the power of entanglement and non-locality, our game tests the power of quantum superposition. In addition, non-local games are device-independent~\cite{Brunner_2014}, meaning that conclusions are drawn solely from the players' output statistics without any reference to the details of their hardware. Our complement sampling game is only partially device-independent, but with a certain rigidity. 

Let us refer to the party responsible for the preparation of the initial state as the referee. The referee is a neutral party and is assumed never to deviate from their intended behavior. However, their state preparation device cannot always be trusted. We have the following result: if the referee's state preparation is trusted, we can certify the quantumness of the player's strategy and device; if instead the player's device and strategy are trusted, we can certify the quality of the referee's state preparation. 
In our context, a device that makes unpredictable random errors while trying to faithfully implement the prescribed operations should be seen as trusted. This, together with the low hardware requirements, makes our test suitable for contemporary noisy hardware.

We now introduce the game under the assumption of a trusted state preparation, then move on to the scenario of a trusted player, and finally demonstrate the game using real quantum hardware. We defer the discussion of potential experimental loopholes to the end of the article.

\section*{Results}

\subsection*{The complement sampling game} 

Consider the set $\omega = \{0,1\}^n$ of bit strings of length $n \in \mathbb{N}$. Given some form of uniform sampling access to a non-empty subset $S \subset \omega$, complement sampling is the problem of returning any sample from the complement set $\bar{S} = \omega \setminus S$.
In particular, given access to the \emph{subset state} $\ket{S} = \frac{1}{\sqrt{|S|}} \sum_{x \in S} \ket{x}$ corresponding to the uniform distribution over the elements of $S$, there is a simple way to solve this problem on a quantum computer~\cite{benedetti2025complement}. The idea is to swap the state to its complement $\ket{S} \rightarrow \ket{\bar{S}}$ and then measure it in the computational basis. Let $\ket{+^n} = \frac{1}{\sqrt{2^n}} \sum_{x \in \omega} \ket{x}$ be the uniform superposition and $I$ be the identity operator. It can be verified that the Grover diffusion operator $U = 2\dyad{+^n} - I$ achieves
\[
U \ket{S} = 2\sqrt{\frac{|S|}{2^n}\left(1-\frac{|S|}{2^n}\right)} \ket{\bar{S}} + \left(2\frac{|S|}{2^n} - 1\right) \ket{S},
\]
independently of the state $\ket{S}$. When $|S| = 2^{n-1}$, i.e. the subset state contains exactly half of the elements, the amplitude of $\ket{S}$ vanishes and one obtains $\ket{\bar{S}}$ deterministically. This works also in reverse, meaning that $U$ is a perfect \emph{swapper circuit}~\cite{benedetti2025complement,aaronson2020hardness}. On the other hand, any classical algorithm accessing a single sample from $S$ and returning a sample from $\bar{S}$ succeeds with probability at most $\frac{2^{n-1}}{2^n -1}$, quickly approaching $1/2$ with increasing bit string length~\cite{benedetti2025complement}. Here we design a single-player game based on complement sampling.

We fix the number of bits $n$ and the number of rounds $r$. Let $\mathcal{P}(\omega)$ denote the power set and let $\mathcal{F} \subset \mathcal{P}(\omega)$ be a family of subsets of fixed cardinality $2^{n-1}$. At each round, the referee draws an instance uniformly at random from the family, $S \sim \mathrm{Unif}(\mathcal{F})$, prepares the subset state $\ket{S}$, and sends it to the player. The player uses any desired strategy to produce a candidate bit string $y$ supposedly from $\bar{S}$. Formally, a strategy is described by a conditional probability distribution $h(y \mid S)$: the probability of outputting $y$ given that the quantum state $\ket{S}$ is received. The candidate bit string is then sent to the referee, who verifies the result of the round. If the bit string belongs to the complement the player scores a point, otherwise they lose a point, i.e. the score is
\[
\sigma(S,y) =  \begin{cases} 
+1,& y\in \bar{S},\\
-1,& y\notin \bar{S}.
\end{cases}
\]

In order to quantify the value of a strategy under a specific family of subsets, we define the \emph{expected score} over the randomness of the input and output
\begin{equation}
\label{eq:bell_func_main}
    V(\mathcal{F}, h) = \mathbb{E}_{S}  \;\mathbb{E}_{y\mid S} \left[\sigma\left(S, y\right)\right] .
\end{equation}
This quantity is akin to a Bell functional~\cite{Buhrman_2012} for our complement sampling game, and is derived in full generality in Supplementary Note 1.
Now let $C$ and $Q$ be the sets of all classical and quantum strategies, respectively. The optimal classical expected score is $V(\mathcal{F}, C) =  \max_{h \in C} \, V(\mathcal{F}, h)$ and the optimal quantum expected score is $V(\mathcal{F}, Q) =  \max_{h \in Q} \, V(\mathcal{F}, h)$. 
We quantify the \emph{non-classicality} of quantum strategies for a specific family $\mathcal{F}$ with the ratio 
\begin{equation}
\label{eq:ratio_shifted_scores}
    \frac{V(\mathcal{F}, Q) - V(\mathcal{F}, h_\mathrm{rnd})}{V(\mathcal{F}, C) - V(\mathcal{F}, h_\mathrm{rnd})} ,
\end{equation}
where $h_\mathrm{rnd}$ is the strategy that guesses $y$ at random. Both numerator and denominator are interpreted as improvements over a completely uneducated guess, and the ratio is the factor by which classicality is violated.
The score function $\sigma$ has been chosen such that $V(\mathcal{F}, h_\mathrm{rnd})=0$, so we drop this term from the ratio hereafter. We will study the ratio as a function of $n$.

To build some intuition, let us first consider a player who does not have access to a quantum computer. The player receives the state $\ket{S}$ via a quantum communication channel, but is not able to perform any coherent operation on it. All the player can do is to measure $\ket{S}$ in the computational basis to obtain a bit string from $S$. This is equivalent to directly receiving a random element of $S$ from the referee, without the need for a quantum communication channel. The player then returns to the referee any bit string different from the observed one. This educated guess may look like a very minor improvement over the uneducated guess. However, under suitable choices of family $\mathcal{F}$, this is provably the optimal classical strategy and achieves an expected score of $\frac{1}{2^n -1}$ at each round. Now let us consider a player who can coherently manipulate $\ket{S}$ to perform a quantum computation. The player implements the swapper circuit to obtain $\ket{\bar{S}}$ from $\ket{S}$, and a measurement in the computational basis yields a bit string from the complement $\bar{S}$. Assuming a noiseless (error-free) quantum setup this strategy is optimal as it achieves score one at each round.

\begin{figure}
\centering
\includegraphics[width=1.0\linewidth]{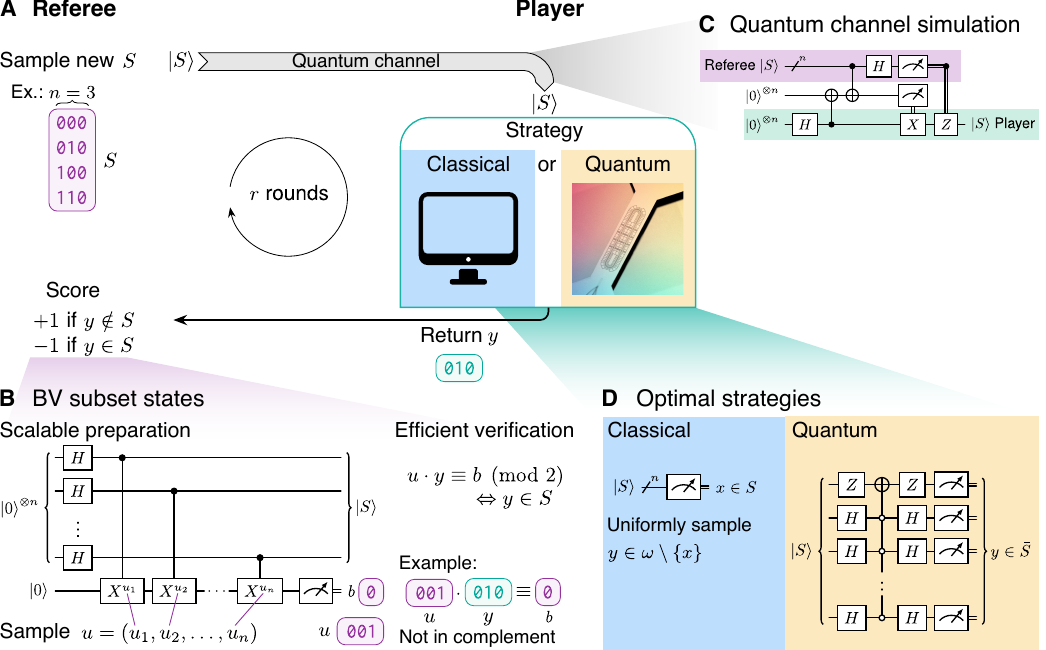}
\caption{\emph{Complement sampling game with trusted state preparation}. (A) At each round, the referee samples a new set $S$ of cardinality $2^{n-1}$, prepares the subset state $\ket{S}$, and sends it to the player. The player chooses a classical or quantum strategy to obtain a sample $y$, supposedly from the complement $\bar{S}$, and sends it to the referee. The referee assigns a score $+1$ if the returned sample is from the complement and $-1$ otherwise. After a number of rounds the referee must determine from the recorded scores if the player used a classical or quantum strategy. (B) Quantum circuit used by the referee to prepare the $n$-qubit Bernstein-Vazirani (BV) subset state $\ket{S}$. Here, $S$ depends on random bit string $u$ and ancilla measurement outcome $b$. The referee verifies $y$ using their knowledge of $u$ and $b$. (C) Teleportation circuit simulating a quantum communication channel between referee and player. This uses $n$ ancillas. (D) The optimal classical strategy is to measure the incoming state in the computational basis and return a different bit string. The swapper circuit implements the unitary $2\dyad{+^n} - I$ and is used by the player adopting the optimal quantum strategy. Additional ancillas can be used to compile this circuit to lower-level gates.}
\label{fig:overview}
\end{figure}

\subsection*{Families of subsets and preparation of random subset states}

The referee must be able to prepare random subset states and verify whether bit strings belong to them. A generic subset state, whose amplitudes are defined by an arbitrary indicator function, is inefficient to prepare due to exponential circuit depth. Ref.~\cite{benedetti2025complement} observes that a uniformly random subset can be obtained from a uniformly random permutation $P$ as
\begin{align}
\label{eq:prp_subsets}
S_P = \Set*{x \given x = P(z) \text{ and the first bit of } z \in \omega \text{ is } 0 } .
\end{align}
This defines a family $\mathcal{F}_\mathrm{PRP}=\Set{S_P}$ containing $\abs{\mathcal{F}_\mathrm{PRP}} = \binom{2^n}{2^{n-1}}$ subsets. A reversible circuit for $P$ applied to the superposition $\ket{0} \otimes \ket{+^{n -1}}$ produces the subset state $\ket{S}$ on a quantum computer. The referee can efficiently verify that $y \in \bar{S}$ by checking whether the first bit of $P^{-1}(y)$ is $1$. When using ensembles of strong pseudorandom permutations (PRPs), complement sampling remains classically hard in the average case~\cite{benedetti2025complement}. This is remarkable because PRPs can be realized in polynomial circuit depth, allowing for a viable implementation of the game. A comprehensive exploration of this proposal is provided in Supplementary Note 4 where we report successful experiments on a variety of quantum hardware settings. 

Unfortunately, a complement sampling game based on PRPs is conditional on the the existence of one-way functions~\cite{katz2020introduction}. A one-way function can be calculated quickly, yet it is conjectured to be computationally intractable to invert. Although modern digital life depends on this conjecture, e.g. for secure communication and banking, the aim of this work is to test aspects of quantum mechanics without relying on computational hardness assumptions. We propose an alternative family of subsets that does not rely on any complexity-theoretical assumption, and is therefore \emph{unconditional}, a significant improvement on the game sketched in~\cite{benedetti2025complement}. Let each pair of $u\in \omega \setminus\Set{0}^n$ and $b\in\Set{0,1}$ define the subset
\[
S_{u,b} = \Set*{ x \in \omega \given u \cdot x \equiv b \pmod 2 },
\]
of cardinality $2^{n-1}$, so that the family $\mathcal{F}_{\mathrm{BV}}=\Set{S_{u,b}}$ contains $\abs{\mathcal{F}_{\mathrm{BV}}}=2(2^n-1)$ subsets. The defining condition is also used in the Bernstein-Vazirani problem~\cite{Bernstein_1997}, where the task is to learn the hidden string $u$ from oracle evaluations of the condition. We refer to such subsets as BV subsets. Importantly, there exists a simple circuit to prepare $\ket{S_{u,0}}$ or $\ket{S_{u,1}}$ each with probability $1/2$, as shown in Fig.~\ref{fig:overview}B. The referee chooses $u$ uniformly at random, stores the value of $b$ obtained during state preparation, and can efficiently verify the player's candidate bit string $y$ by checking that $u \cdot y \equiv 1 - b \pmod 2$. 

In Supplementary Note 1 we prove that the optimal classical expected score for BV subsets vanishes exponentially with the bit string length, $V(\mathcal{F}_{\mathrm{BV}}, C) = \frac{1}{2^n -1}$. In contrast, the ideal (error-free) quantum algorithm for complement sampling on these subsets gives $V(\mathcal{F}_{\mathrm{BV}}, Q) = 1$. It follows that the ideal complement sampling game based on BV subsets exhibits an unconditional, exponentially large violation of classicality quantified by
\[
\frac{V(\mathcal{F}_{\mathrm{BV}}, Q)}{V(\mathcal{F}_{\mathrm{BV}}, C)} = 2^n -1 .
\]

We emphasize that this result holds because each round uses one copy of a randomly selected BV subset state. In Supplementary Note 2 we show that had we provided $\mathcal{O}(n)$ copies per round, a classical strategy would win with high probability thus achieving a high score. However, a multi-copy version of the game based on the PRP subsets in Eq.~\eqref{eq:prp_subsets} would still exhibit a large \emph{conditional} violation of classicality. This is because a classical strategy cannot distinguish PRP subsets from random ones given access to any polynomial number of copies.

\subsection*{Playing without trust in the state preparation}

We have so far assumed that the referee's device achieves a perfect state preparation, while the player could play any strategy they want and thus even use an untrusted device. This however may prevent us from making conclusive statements when, instead, it is the state preparation that cannot be trusted. Complement sampling exhibits an interesting rigid structure which we exploit next. We slightly modify both the game and the score such that they become sensitive to the state preparation.

At each round the referee wants to prepare $\ket{S}$, but the device prepares an unknown state $\rho_S$ instead. The referee flips an unbiased coin $c \in \{0,1\}$: if $c = 0$ the round proceeds as normal, i.e. the player receives $\rho_S$ and uses strategy $h$ to output a bit string $y$, scoring $\sigma(S,y)$; if $c = 1$ the referee measures $\rho_S$ in the computational basis, obtains a bit string $y$, and counts a score of $- \sigma(S,y)$. 
The \emph{rigidity score} is defined as
\[
    R(\mathcal{F}, h) = \mathbb{E}_{S}\mathbb{E}_{c} \mathbb{E}_{y|S,c} \left[(-1)^c \, \sigma\left(S, y\right) \right].
\]

In Supplementary Note 3, we prove that when a trusted player uses the optimal quantum strategy, then 
\[
R(\mathcal{F}, Q) \geq 1 - \epsilon \; \implies \; \mathbb{E}_S[\bra{S}\rho_S \ket{S}] \geq 1 - \epsilon ,
\]
where the right-hand side is the expected fidelity of state preparation. Up to an additional factor of $2$ in front of $\epsilon$, the converse also holds true.

The rigidity score retains the properties of the expected score discussed in the previous sections. Namely, when the state preparation is trusted, the rigidity score can be used to quantify the non-classicality of a strategy as a ratio of shifted scores akin to Eq.~\eqref{eq:ratio_shifted_scores}. Specializing to the family of BV sets, we show in Supplementary Note 3 that the optimal classical strategy satisfies $R(\mathcal{F}_\text{BV}, C) = \frac{1}{2}(1 + \frac{1}{2^n -1})$, while $R(\mathcal{F}_{BV}, Q) = 1$ for the optimal quantum strategy.

Finally, if both the referee's and player's devices are simultaneously untrusted the results of the game based on BV subsets are inconclusive. This is proven in Supplementary Note 3 where we let the devices collude to achieve a perfect score without using a quantum computer. Note that collusion is also possible in Bell tests and non-local games with an untrusted referee's device.  
We summarize the rigidity score under different trust assumptions in Table~\ref{tab:scores_with_untrusted_party}.

\begin{table}
\centering
\small
\renewcommand{\arraystretch}{1.4}
\begin{adjustbox}{center}
\begin{tabular}{|c|c|c|c|}
\cline{3-4}
\multicolumn{2}{c|}{} & \multicolumn{2}{c|}{Referee's state preparation device} \\
\cline{3-4}
\multicolumn{2}{c|}{} & Trusted & Untrusted \\
\hline
\multirow{2}{*}{\makecell{Player's optimal\\ classical strategy}}
  & Trusted
  & $R = \frac{1}{2}\!\left(1 + \frac{1}{2^n - 1}\right)$
  & $R \leq \frac{1}{2}\!\left(1 + \frac{1}{2^n - 1}\right)$\\
\cline{2-4}
  & Untrusted
  & $R \leq \frac{1}{2}\!\left(1 + \frac{1}{2^n - 1}\right)$
  & $R=1$\\
\hline
\multirow{2}{*}{\makecell{Player's optimal\\ quantum strategy}}
  & Trusted
  & $R=1$
  & $R \geq 1 - 2\epsilon$\\
\cline{2-4}
  & Untrusted
  & $R \leq 1$
  & $R=1$\\
\hline
\end{tabular}
\end{adjustbox}
\caption{\emph{Rigidity score for the game with BV subsets under various assumptions about the devices.} First column: as the trusted state preparation sends $\ket{S}$ to the player, the provably optimal classical and quantum strategies provide tight upper bounds to the score. Second column: the untrusted state preparation yields a state $\rho_S$ with expected fidelity $\mathbb{E}_S[ \bra{S} \rho_S \ket{S}] \geq 1 - \epsilon$ . If both devices are untrusted, they can collude to achieve a prefect score even classically and the game in not conclusive.}
\label{tab:scores_with_untrusted_party}
\end{table}

\subsection*{Experiments on quantum hardware}

From the circuits in Fig.~\ref{fig:overview} it is clear that an efficient and scalable implementation of our game requires flexible qubit-to-qubit connectivity, mid-circuit measurements, real-time classical processing for conditional operations, and high-fidelity quantum gates. We choose to demonstrate the game on Quantinuum System Model H2 trapped-ion quantum computers as they provide all these features. These systems comprise a 56-qubit quantum charge-coupled device (QCCD) architecture where $^{171}\text{Yb}^+$ qubit ions and $^{138}\text{Ba}^+$ coolant ions are shuttled in a race-track trap for dynamic operations~\cite{DeCross_2025,Moses_2023}. Physical qubits are transported to four dedicated interaction zones where lasers execute high-fidelity quantum operations, enabling all-to-all connectivity and parallelized gate execution. 
At the time of all but the experiments for $n=37$, H2-2 achieved two-qubit gate infidelity of $(1.10 \pm 0.17) \times 10^{-3}$ and memory error per qubit and depth-1 circuit time of $(2.26 \pm 0.29) \times 10^{-4}$ measured via randomized benchmarking of all qubits~\cite{Moses_2023,Magesan_2011,Harper_2019,hseries_benchmarking}.
At the time of the experiments for $n=37$, H2-2 achieved two-qubit infidelity of $(8.30 \pm 0.48) \times 10^{-4}$ and memory error per qubit and depth-1 circuit time of $(1.20 \pm 0.20) \times 10^{-4}$, making it one of the highest fidelity quantum processors available. Yet due to these dominant sources of error we are limited to executing a few hundreds of gates before the signal will become indistinguishable from the classical score at our chosen significance level. We emphasize that we cannot use standard error mitigation techniques because they require copies of the output state~\cite{Cai_2023}. Copies are not considered in our current formulation of the game where each round requires the preparation of a different subset state $\ket{S}$. Instead we reduce gate error by optimizing the number of two-qubit gates using compilation tools.

In the absence of a true quantum communication channel connecting quantum computers, we choose to embed both referee's and player's qubit registers within a single quantum processor. The referee's state preparation circuit requires one ancilla  (Fig.~\ref{fig:overview}B). The communication channel is simulated using the quantum teleportation protocol~\cite{Bennett_1993}. We prepare a set of $n$ Bell pairs so that $n$ additional qubits are in the referee's ancilla register and $n$ in the player's main register (Fig.~\ref{fig:overview}C). We utilize mid-circuit measurements and conditional gates to implement the corrections required by the teleportation protocol. Finally, we allow the player to use additional ancilla qubits for the implementation of the swapper circuit (Fig.~\ref{fig:overview}D). Our optimized implementation based on Ref.~\cite{Maslov_2016} requires $\ceil*{ \frac{n - 3}{2}}$ ancillas and achieves a linear scaling of the number of two-qubit gates, compared to the quadratic scaling of textbook implementations (see Supplementary Note 5 for details). In total our experiments with BV subsets use $n_\mathrm{qubits} = 3n + \ceil*{\frac{n - 3}{2}} + 1$ qubits. System Model H2 quantum computers have 56 qubits. This limits the bit string length to $n\leq 16$. To go beyond, we perform experiments without teleportation. These require $n_\mathrm{qubits} = n + \ceil*{\frac{n - 3}{2}} + 1$ qubits and are thus limited to $n \leq 37$. Importantly, our compilation respects the barriers between referee, teleportation and player circuits, as to minimize leakage of information.

In order to produce truly random bit strings $u$ for the preparation of BV subsets, we employ the Quantum Origin solution by Quantinuum. In Quantum Origin a two-source extractor combines the local randomness source and the outcomes of a Bell test to generate a seed. This seed is used in a seeded extractor which recursively distills near-uniform random bits from fresh blocks of data from the local source of randomness~\cite{Foreman_2023}. This solution provides random bits with at most $2^{-128}$ total variation distance from the uniform distribution. For our purposes, the Bell tests were performed on the H1-1 quantum computer.

\begin{table}
\centering
\begin{tabular}{c|c|c|c|c} 
$n$ & Teleportation & $n_\mathrm{qubits}$ & $n_\mathrm{gates}$ avg & $n_\mathrm{gates}$ std \\ 
\hline 
5   & yes & 17 & 29.56  & 1.16 \\ 
10  & yes & 35 & 72.16  & 1.50 \\ 
15  & yes & 52 & 114.95 & 1.87 \\ 
20  & no  & 30 & 117.46 & 2.08 \\ 
25  & no  & 37 & 149.83 & 2.46 \\
30  & no  & 45 & 182.16 & 2.60 \\ 
35  & no  & 52 & 214.69 & 2.95 \\ 
37  & no  & 55 & 227.61 & 3.11 
\end{tabular}
\caption{\textit{Summary of circuit sizes for the player's experiments.} Each row corresponds to a complement sampling game the player executed on the 56-qubit H2-2 quantum computer for $r/2$ rounds, where each round utilizes its own random subset $S_{u,b}$. This amounts to running 1200 different circuits and obtaining a single sample from each. 
The average number of native two-qubit gates $n_\mathrm{gates}$ (arbitrary angle ZZ rotations) and their standard deviation are estimated from the $r/2$ samples.
} 
\label{tab:experiment_setup_bv} 
\end{table}

\begin{figure}
\centering
\includegraphics[width=\linewidth]{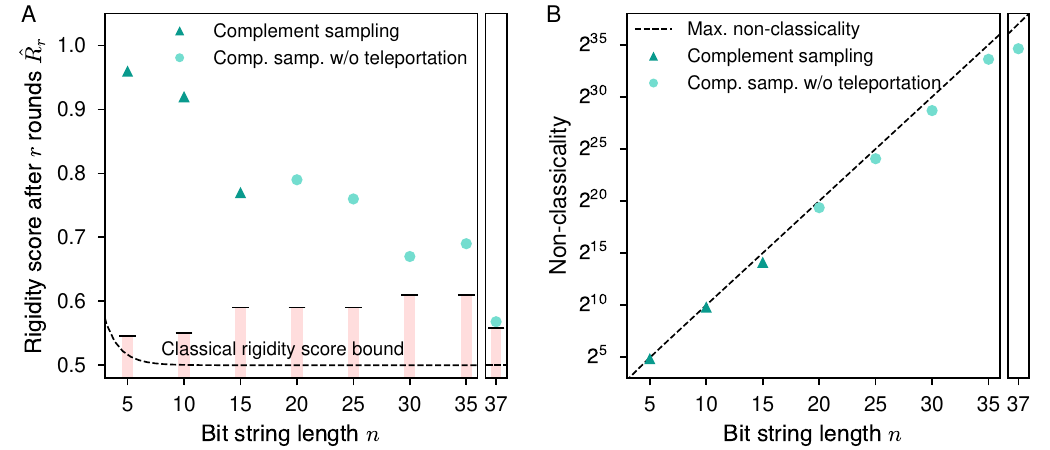}
\caption{\textit{Experimental results.} (A) Empirical rigidity score for the complement sampling game executed on the H2-2 quantum computer. Outside the shaded region we reject the null hypothesis (that the game was played with any strategy with rigidity score at most the optimal classical score bound) with significance $\alpha=0.01$. The $n=37$ data use $r=1000$ rounds, all other data use $r=200$ rounds. (B) Violation of classicality quantified by the ratio of the expected scores for quantum vs classical strategies, $\hat{V}_r/V(\mathcal{F}_\text{BV}, C)$.}
\label{fig:experiment_results_bv}
\end{figure}

We opt for a game with trust in the state preparation, cf. Table~\ref{tab:scores_with_untrusted_party}, and choose a fixed number of rounds $r = 200$ ($r=1000$ for our largest experiments) and a significance level $\alpha=0.01$. We also fix the number of times the player plays their strategy ($c=0$) and the referee measures and scores the prepared BV state ($c=1$) to $r/2$ each. This means we simulate coin flips $c$ with equal number of outcomes for $c=0$ and $c=1$.
To score the state preparation, we measure the referee's register in the computational basis immediately after the state preparation step (Fig.~\ref{fig:overview}B). This produces bit strings $y^{(i)}$ that we use to compute the empirical score $\hat{P}_r = \frac{2}{r} \sum_{i=1}^{r/2} -\sigma(S^{(i)}, y^{(i)})$. In the noiseless setup and if the referee prepared the correct subset state, $y\in S$ yielding an ideal expected score of $1$.

The player implements the quantum strategy in each round. We estimate the expected score of the player's strategy, Eq.~\eqref{eq:bell_func_main}, from their $r/2$ rounds of the game as $\hat{V}_r = \frac{2}{r} \sum_{i=1}^{r/2} \sigma(S^{(i)}, y^{(i)})$. The empirical rigidity score over $r$ rounds is $\hat{R}_r = \hat{P}_r + \hat{V}_r$. Circuit sizes for the full player's experiments are reported in Table~\ref{tab:experiment_setup_bv}. Figure~\ref{fig:experiment_results_bv}A shows the observed scores. As expected, $\hat{R}_r$ decreases with the number of qubits and gates but remains significantly above the classical rigidity score bound $R(\mathcal{F}_\mathrm{BV}, C)$. 

To quantify significance of those results we now formulate a hypothesis test under the assumption that the state preparation is trusted, but a similar test can be formulated if the player's device is trusted instead. Consider the null hypothesis: the player adopts a strategy with rigidity score less than or equal to the optimal classical strategy, $\leq R(\mathcal{F}_{\mathrm{BV}}, C)$, in each round. The alternative hypothesis is: the player adopts a non-classical strategy with expected rigidity score $> R(\mathcal{F}_{\mathrm{BV}}, C)$. Denote the score of each round $i = 1, \dots, r$ as an i.i.d. random variable $R_i$ with expectation $\mathbb{E}[R_i]$. We now run the game and record the empirical mean score $\hat{R}_r$ after $r$ rounds.
Let $t = \hat{R}_r - R(\mathcal{F}_{\mathrm{BV}}, C)$ be the difference between the observed mean score and the best possible classical rigidity score. If $t\leq 0$, we fail to reject the null hypothesis. If $t>0$, we compute the p-value
\[
\mathbb{P}\left( \frac{1}{r}\sum_{i=1}^r R_i - \mathbb{E}[R_i]  \geq t \right) = p
\]
via bootstrapping. We use the bias-corrected and accelerated bootstrap method~\cite{efronIntroductionBootstrap1994} implemented in \texttt{SciPy} drawing 99,999 samples with replacement from the set $\Set{R^{(i)} - \hat{R}_r + R(\mathcal{F}_{\mathrm{BV}}, C)}_{i=1}^r$. This creates an empirical distribution of mean estimators $\hat{R}_{r,0}$ assuming the null hypothesis is true. The resulting p-value is the probability of obtaining a statistic at least as extreme as the one observed, assuming the null hypothesis is true.
Allowing a one-sided type I error (incorrectly rejecting the null hypothesis) probability of $\alpha$ we thus reject the null hypothesis if $p \leq \alpha$.
The p-values resulting from the observations in Fig.~\ref{fig:experiment_results_bv}A are: $p=2.6\times 10^{-4}$ ($n=30$), $p=6.9\times 10^{-6}$ ($n=35$), $p=2.5\times 10^{-3}$ ($n=37$). All other experiments resulted in an estimate $p=0$ because no resampled score under the null hypothesis was higher than the observed score. Since $p \leq \alpha = 0.01$, we reject the null hypothesis that the data came from a strategy with expected rigidity score $\leq R(\mathcal{F}_\mathrm{BV}, C)$ in each round (this includes all classical strategies) with significance $0.01$, corroborating the non-classical nature of the quantum device. The shaded area in Fig.~\ref{fig:experiment_results_bv}A shows the region where we cannot reject the null hypothesis with significance $\alpha$. The upper limit of the shaded region is the $1-\alpha$ quantile of the bootstrapped distribution. A simple way to reduce the non-rejection region is to run the game for more rounds. This is illustrated for $n=37$, where we ran the game for $r=1000$ instead of $r=200$ rounds resulting in a tighter non-rejection region than, e.g., for $n=35$. 
It is interesting to note that $n=20$ performs better than $n=15$ thanks to a simpler circuit (w/o teleportation). Furthermore $n=35$ performed better than $n=30$ due to a higher state preparation empirical score $\hat{P}_r$. 

Figure~\ref{fig:experiment_results_bv}B shows the ratio $\hat{V}_r \, / \, V(\mathcal{F}_{\mathrm{BV}}, C)$ versus bit string length, i.e. the non-classicality observed in our experiments. We observe non-classicality that grows exponentially with the bit string length, closely tracking the non-classicality of the optimal quantum strategy. Deviations from the optimal strategy due to noise become clear for large number of two-qubit gates, e.g. at $n=37$. 

\section*{Discussion}

We proposed to test the quantum mechanical nature of hardware via a game of complement sampling. In the idealized setup where the referee's state preparation device is trusted, conclusions are drawn solely from the analysis of the player's output statistics without having to specify the inner working of their device and measurements. This is the hallmark of device-independent protocols.
One could thus have a referee with a trusted device testing a large number of players. 
When the referee's device cannot be trusted, but the player is trusted, we introduced a variation of the game that can be used to certify the quality of state preparation. 
This role inversion is possible thanks to a rigid structure in complement sampling which we uncovered in this work.
An explicit protocol for certifying the state preparation has not been given in this work, as we focused on testing non-classicality of the player instead.
If both referee and player use untrusted devices, these can collude to spoof the results. We note that, if the referee's setup deviates from the intended behavior, then collusion can also occur in Bell tests and non-local games.

We demonstrated the game of complement sampling on Quantinuum H2 quantum computers. The statistics generated by our experiments exhibit exponentially large violations of classicality (of the order of $2^n$ up to $n=37$ on H2-2). These could be explained by a systematic adoption of a near-optimal quantum strategy under an optimal state preparation, or by an optimal quantum strategy under a near-optimal state preparation. This result is a further substantiation of the quantum nature of the hardware used, and offers a different perspective than that of non-local games, quantum volume ($2^{25}$ for H2-2~\cite{hseries_quantumvolume}) and application-oriented benchmarks. In particular, our test demonstrates the power of quantum superposition in a manner that is oblivious of entanglement and non-locality. It is important to note that our current demonstration has several loopholes that shall be resolved in future work. For example, we used quantum teleportation to simulate a communication channel, yet both referee and player manipulated the subset states on the same H2 computer. We also used the H1 computer and Quantum Origin to generate random bits for the referee. To avoid any information leakage, the referee should generate both random numbers and subset states using their own computer. The subset state should then be sent to the player's own computer over an actual quantum communication channel.

The interested reader can find more experiments and analyses in Supplementary Note 4. 
There we utilize reversible circuits for approximate $k$-wise independent pseudorandom permutations (PRPs)~\cite{gay2025pseudorandomness} to construct the subset states for our game. The approximation error can, in principle, be exploited to improve a classical strategy, while hardware noise penalizes the quantum strategy more due to the larger number of two-qubit gates. Nevertheless, these secondary results provide further evidence of large violations of classicality on the H2 quantum computers, albeit conditional on the existence of one-way functions. Although PRP subsets do not enjoy the nice properties of the Bernstein-Vazirani (BV) subsets used in our main experiments, we expect them to become useful in a variation of the game where the referee provides $k$ copies of the subset state to the player at each round. While this poses a threat to the security of the hidden bit string defining the BV subsets (see Supplementary Note 2), it does not affect the security of $k$-wise independent PRPs because they are indistinguishable from truly random permutations up to the $k$-th moment.

How can we scale the tests and experiments presented in this work to larger sizes? Two-qubit gate fidelities of future hardware are constrained by the physical limits of the underlying technology. The next generation of quantum hardware may thus enable experiments on slightly larger instances without requiring any modifications. To go beyond that, fault-tolerance becomes important and quantum error correcting codes must be considered. These require the compilation of our circuits in terms of a suitable, reduced gate set such as Clifford + T. Our preparation of BV subset states requires only CX and Hadamard gates, which are Clifford gates with a cheap transversal implementation in some codes, such as color codes~\cite{Bombin_2006}. On the other hand, our swapper circuit is essentially a multi-qubit Toffoli gate, requiring T gates expensively implemented by, e.g. magic state distillation~\cite{Bravyi_2005}. Ref.~\cite{gosset2025multiqubit} (see also relevant earlier work in Ref.~\cite{Hoyer_2005}) gives an approximate implementation of the $n$-qubit Toffoli gate where the number of T gates is independent of $n$, presenting a potential avenue for implementations of the complement sampling game on error-corrected architectures. With regards to the randomness consumed during experiments, should a quantum  generator not be available, the referee may resort to well-established software for pseudorandom number generation. We currently do not know if non-uniform sampling of subset states and coin flips would have negligible effects or otherwise.

Finally, we leave open the question of whether there exists a game with super-exponential violations of classicality.

\paragraph{Data and code availability.} The data and code that support the findings of this study are available at Zenodo~\cite{zenodo_complement_sampling2025}.

\paragraph{Acknowledgements.} We thank Nick Van Duyn, Cameron Foreman and Duncan Jones for their support with the Quantum Origin random number generator. We also thank David Amaro, Cameron Foreman and Konstantinos Meichanetzidis for helpful discussions. 

\paragraph{Author contributions.} M.B., J.W. and H.B. conceived the study and performed analytic calculations. G.M-S. and M.R. wrote the code, performed numerical simulations and hardware experiments, and produced the figures. All authors analyzed the data, interpreted the results and wrote the manuscript.

\clearpage

\providecommand{\href}[2]{#2}\begingroup\raggedright\endgroup

\clearpage

\setcounter{figure}{0}
\setcounter{table}{0}
\makeatletter
\renewcommand{\fnum@figure}{Supplementary Figure \thefigure}
\renewcommand{\fnum@table}{Supplementary Table \thetable}
\makeatother
\renewcommand{\thesection}{Supplementary Note \arabic{section}}
\titleformat{\section}[block]{\normalfont\large\bfseries}{\thesection:}{.5em}{}{}
\renewcommand{\cfttoctitlefont}{\normalfont\large\bfseries}
\renewcommand{\cftsecfont}{\normalfont}
\renewcommand{\cftsecpagefont}{\normalfont}
\renewcommand{\cftsecleader}{\cftdotfill{\cftdotsep}}
\renewcommand{\cftsecaftersnum}{:} 
\addtolength{\cftsecnumwidth}{98pt} 

\begin{center}
\LARGE Supplementary Information for \\
``Unconditional and exponentially large violation of classicality''\\
\vspace{15px}
\large Marcello Benedetti, Gabriel Marin-Sanchez, Jordi Weggemans,\\
Matthias Rosenkranz, and Harry Buhrman
\end{center}

\vspace{15px}

\tableofcontents

\vspace{15px}

\section{Bell functional and violation of classicality}
\label{app:proof_max_violation}

The complement sampling game is a single-player game that does not rely on entanglement as a resource. Regardless, it is insightful to define the game analogously to non-local games, and equip it with Bell-like inequalities. We define a Bell functional~\cite{Buhrman_2012} for the game as the expected score over the randomness of input and output, i.e. $\mathbb{E}_\text{input} \mathbb{E}_{\text{output}\mid \text{input}}[\text{score}]$. Importantly, the score function is chosen so that if the player were to play completely at random, the expected score would be zero. This allows us to compare strategies in terms of their advantage relative to random guessing. 

\paragraph{General complement sampling functional.}
Let $\omega$ be a non-empty and finite universe, and $\mathcal{P}(\omega)$ be its power-set. The referee draws an instance $S \in \mathcal{P}(\omega) \setminus \Set{ \varnothing, \omega }$ with probability $p(S)$ and provides the player with the subset state $\ket{S} = \frac{1}{\sqrt{|S|}} \sum_{x \in S} \ket{x}$ as input. The player adopts any strategy they want and outputs a bit string in $\omega$. We define a strategy as a function $h: \Sigma^n \to \Gamma^n$, from the set of quantum states over $n$ qubits, $\Sigma^n$, to the set of probability vectors over $n$ bits, $\Gamma^n$. We use $h(y \mid S)$ to indicate the probability of output $y$ given input $\ket{S}$. The Bell functional value of this game is the expected score
\begin{equation}
\label{eq:bell_func_supp}
    V(p, h) = \sum_{S \in \mathcal{P}(\omega) \setminus \Set{ \varnothing, \omega }} p(S) \; \sum_{y \in \omega} h(y \mid S) \;\sigma\left(S, y \right) ,
\end{equation}
where the $\pm 1$ score is 
\begin{equation}
\label{eq:plus_minus_score}
    \sigma(S,y) = \begin{cases} 
    +1,& y\in \bar{S},\\
    -1,& y\notin \bar{S}.
    \end{cases}
\end{equation}
Equivalently, $\sigma(S,y)=2 \iv{y\in \bar{S}} - 1$ where $\iv{ \cdot }$ is the Iverson bracket. For the purpose of showing an exponentially large violation of classicality of our game, we can restrict $p$ to be the uniform distribution over a family of subsets $\mathcal{F} \subset P(\omega)$. We denote the corresponding value of the strategy with $V(\mathcal{F}, h)$. Given a set of strategies $H$, we define its value as
\[
    V(\mathcal{F}, H) = \sup_{h \in H} V(\mathcal{F}, h).
\]
We are interested in computing the ratio
\[
    \frac{V(\mathcal{F}, Q) - V(\mathcal{F}, h_\mathrm{rnd})}{V(\mathcal{F}, C) - V(\mathcal{F}, h_\mathrm{rnd})} ,
\]
for specific families $\mathcal{F}$, quantum strategies $Q$, and classical strategies $C$. 
Here $h_\mathrm{rnd}$ is the strategy that guesses $y$ completely at random, and since $V(\mathcal{F}, h_\mathrm{rnd}) = 0$, we hereafter drop it from the ratio. This is without loss of generality since the ratio above is invariant to affine transformations of the score $\sigma' = a \sigma + b$.

A quantum strategy can manipulate $\ket{S}$ coherently, and perform arbitrary measurements to produce the output distribution. In contrast, a classical strategy begins with a measurement of $\ket{S}$ in the computational basis, then processes the result on a classical computer. In this latter case, the input model is equivalent to directly receiving one bit string sampled uniformly from $S$.

\paragraph{Special case: Bernstein-Vazirani subset states.}
Let us specialize the above definitions  to Bernstein-Vazirani (BV) subsets:
\begin{align*}
    &\omega = \Set{0,1}^n \qquad &&\abs{\omega} = 2^n \qquad \\
    &\mathcal{F}_{\mathrm{BV}} = \Set*{ S_{u,b} \given u\in \omega \setminus\Set{0}^n \text{ and } b\in\Set{0,1} } \qquad &&\abs{\mathcal{F}_{\mathrm{BV}}} = 2(2^n -1) \\
    &S_{u,b} = \Set*{ x \in \omega \given u \cdot x \equiv b \pmod 2 } \qquad && \abs{S_{u,b}} = 2^{n -1} .   
\end{align*}
We let the input follow a uniform probability distribution over $\mathcal{F}_{\mathrm{BV}}$
\begin{align}
\label{eq:subset_distribution}
p(S) = \begin{cases} 
    \frac{1}{2(2^n -1)} , & S \in \mathcal{F}_{\mathrm{BV}},\\
    0, & \text{otherwise} .
\end{cases}
\end{align}

\paragraph{Quantum value.}
Since any BV subset contains half the elements of the universe, a simple application of the swapper circuit gives the complement subset state $\ket{\bar{S}} = (2\dyad{+^n} - I) \ket{S}$~\cite{benedetti2025complement}. A measurement in the computational basis gives a bit string from the complement. This quantum strategy implements the conditional probability distribution
\begin{align}
\label{eq:quantum_strategy}
h(y  \mid S) = \begin{cases} 
    \frac{1}{2^{n-1}} , & y \in \bar{S} \\
    0, & \text{otherwise} .
    \end{cases}
\end{align}
It is easy to verify that in the quantum case plugging in Eqs.~\eqref{eq:subset_distribution} and~\eqref{eq:quantum_strategy} in Eq.~\eqref{eq:bell_func_supp} gives an expected score of one, which is optimal. Thus for the set of all quantum strategies $Q$ we have  $V(\mathcal{F}_{\mathrm{BV}}, Q) = 1$. 

\paragraph{Classical value.} 
The classical strategy begins by measuring the subset state in the computational basis, thus reading off a bit string $x$ with Born probability  
\begin{align}
\label{eq:born_probability}
     h(x \mid S) = |\braket{x}{S}|^2 = \begin{cases} 
        \frac{1}{2^{n-1}} , & x \in S \\
        0, & \text{otherwise} .
        \end{cases}
\end{align}
Then the strategy outputs a bit string $y$ with probability $h(y \mid x)$. For example, consider a function $s$ that outputs a particular bit string $s(x) \neq x$ (i.e. it does not have any fixed point). Such function leads to a deterministic output 
\begin{align}
\label{eq:deterministic_step}
h(y \mid x) = \iv{y = s(x)} .
\end{align}
The two steps just described implement
\begin{align}
    h(y \mid S) &= \sum_{x \in \omega} h(y \mid x) \, h(x \mid S) \notag \\
    &= \frac{1}{2^{n-1}} \sum_{x \in S} \iv{y = s(x)}, \label{eq:classical_strategy}
\end{align}
where we marginalize over the measurement outcomes and use conditional independence $h(y|x,S) = h(y|x)$ since, under the assumption of a trusted referee, no information about the subset could have leaked. We now analyze this classical strategy, and later we will justify why it is essentially the optimal one.
Plugging Eqs.~\eqref{eq:subset_distribution} and~\eqref{eq:classical_strategy} in Eq.~\eqref{eq:bell_func_supp} we obtain the expected score
\begin{align*}
    V(\mathcal{F}_{\mathrm{BV}}, h) &= \sum_{S \in \mathcal{F}_{\mathrm{BV}}} \frac{1}{2(2^n -1)} \, \sum_{y \in \omega} \frac{1}{2^{n-1}} \, \sum_{x \in S} \iv{y = s(x)} \sigma\left(S, y \right) \\
    &= \frac{1}{(2^n -1) 2^{n}} \sum_{S \in \mathcal{F}_{\mathrm{BV}}} \, \sum_{x \in S}  \sigma\left(S, s(x) \right) \\
    &=\frac{1}{(2^n -1) 2^{n}} \sum_{u \in \omega \setminus \Set{0}^n} \, \sum_{ b \in \Set{0,1}} \, \sum_{x \in S_{u,b}}  \sigma\left(S_{u,b}, s(x) \right) \\
    &=\frac{1}{(2^n -1) 2^{n}} \sum_{u \in \omega \setminus \Set{0}^n} \, \sum_{x \in \omega}  \sigma\left(S_{u,u\cdot x}, s(x) \right) .
\end{align*}
In the third line we express the subsets in terms of their parameters $u$ and $b$. In the fourth line we use that for fixed $u$, any $x$ is either in $S_{u,0}$ or in $S_{u,1}$ depending on the result of $u\cdot x \pmod 2$. Next we rewrite the score as  
\[
    \sigma\left(S_{u,u \cdot x}, s(x) \right) = 2 \iv{u\cdot s(x)\neq u\cdot x} - 1 ,
\]
so that 
\begin{align}
\label{eq:intermediate_step}
    V(\mathcal{F}_{\mathrm{BV}}, h) = \Bigg(\frac{2}{(2^n -1) 2^{n}} \sum_{x \in \omega} \sum_{u \in \omega \setminus \Set{0}^n} \iv{u \cdot s(x)\neq u\cdot x}\Bigg) - 1.
\end{align}
We now count how many $u$'s satisfy the condition $u\cdot s(x)\neq u\cdot x$ or, equivalently, $u\cdot (s(x) \oplus x) =1 $. Recall that we assumed a map without fixed points, i.e. $s(x) \neq x$ for all $x$. Letting $d(x):=s(x)\oplus x$, our assumption gives $d(x) \neq \Set{0}^n$ for all $x$. For fixed $x$, exactly half of the possible bit strings $u$ satisfy $u \cdot d(x)=1$. Therefore each $x$ contributes $2^{n}/2$ to the summation in Eq.~\eqref{eq:intermediate_step}. We arrive at
\[
    V(\mathcal{F}_{\mathrm{BV}}, h) = \Bigg(\frac{2}{(2^n -1) 2^{n}} \sum_{x\in \omega} \frac{2^{n}}{2}\Bigg) - 1 = \frac{2^n}{2^n -1} - 1 = \frac{1}{2^n - 1}.
\]
We remark that: i) any fixed-point-free map $s$ (e.g.\ $s(x) = x\oplus e_1$ for a fixed nonzero $e_1$) attains this value; ii) maps with fixed points $s(x) = x$ can only perform worse since $x$ is always sampled from $S$ following Eq.~\eqref{eq:born_probability}; and iii) randomization cannot improve the value since each $x$ contributes the same amount regardless of which $s(x)\neq x$ is chosen. We conclude that the optimal expected score over classical strategies is
\begin{align}
      V(\mathcal{F}_{\mathrm{BV}}, C) = \max_{h \in C} V(\mathcal{F}_{\mathrm{BV}}, h) = \frac{1}{2^n - 1} 
\label{eq:quantum_score_BV_states}
\end{align}
and the maximum is attained by the classical strategy described above.

\paragraph{Quantum-classical ratio.}
Using the above results for BV subset states we arrive at the ratio of the expected scores of quantum and classical strategies
\[
    \frac{V(\mathcal{F}_{\mathrm{BV}}, Q)}{V(\mathcal{F}_{\mathrm{BV}}, C)} = 2^n - 1 .
\]
Our complement sampling game gives an unconditional and exponentially large violation of classicality.

\section{Classical sample complexity bounds for Bernstein-Vazirani subsets}
\label{app:upper_bound_BV_states}

For the subset state family $\mathcal{F}_{\mathrm{BV}}$, we will show in this Supplementary Note that the classical randomized sample complexity to solve complement sampling is $\Theta(n)$.

\paragraph{Upper bound.} The key idea of our upper bound procedure is that, for any fixed $x_0 \in S$ and any uniformly random $x_j \in S$, the vector $d_j = x_j \oplus x_0$ is uniformly random in the $(n-1)$-dimensional subspace $V = \Set{v \in \mathbb{F}_2^n \given v \cdot u \equiv b \pmod 2}$. If the matrix $D$ whose rows are the vectors $\Set{d_j}$ spans $V$, then $\ker (D) = \operatorname{span} \Set{u}$. Thus $u$ is the unique nonzero null vector and it can be found by Gaussian elimination. Once we have found $u$, we can use $x_0$ to check whether $b = 0$ or $b = 1$, and then sample uniformly random bit strings $y$ until we find one that satisfies $y \cdot u \equiv 1 - b \pmod 2$. This last step requires only two attempts in expectation and does not cost new input samples (we sample from $\Set{0,1}^n$ and not from $S$). Returning $y$ solves complement sampling with $\mathcal{O}(n)$ sample complexity.

\begin{theorem}
Fix any $S \in \mathcal{F}_{\mathrm{BV}}$. For any $\epsilon>0$, there is a classical procedure that uses
\[
m = n +\bigl\lceil \log (1/\epsilon) \bigr\rceil
\]
i.i.d.\ samples from $S$ (with replacement) and recovers the hidden $u \in \Set{0,1}^n \setminus \Set{0}^n$ with probability at least $1-\epsilon$. The runtime is polynomial in $n$ and $\log(1/\epsilon)$.
\end{theorem}

\begin{proof}
Let $u \in \Set{0,1}^n \setminus \Set{0}^n$ and $b \in \Set{0,1}$ be such that
\[
S = \Set{x \in \Set{0,1}^n \given x \cdot u \equiv b \pmod 2}.
\]
Repeatedly measuring different copies of $\ket{S}$ in the computational basis is equivalent to sampling $x_0, x_1, \ldots, x_{m-1}$ i.i.d.~uniformly from $S$ with replacement. For $j \in \Set{1,\dots,m-1}$, set
$d_j := x_j \oplus x_0$ and form the $(m-1) \times n$ matrix $D$ whose rows are $d_1,\dots,d_{m-1}$. For any $x_0, x \in S$ we have $(x_0 \oplus x)\cdot u \equiv 0 \pmod 2 $.
Hence each $d_j$ lies in the subspace
\[
V = \Set{ v \in \mathbb{F}_2^n \given v \cdot u \equiv 0 \pmod 2 }, \qquad \dim V = n-1.
\]
Because the map $x \mapsto x_0 \oplus x$ is a bijection and the $x_j$ are i.i.d.\ uniform in $S$, the $d_j$ are i.i.d.~uniform in $V$.
Thus $\operatorname{rowspan}(D) \subseteq V$, so $\rank(D) \leq n-1$. By the rank--nullity theorem,
\[
\rank(D) + \operatorname{nullity}(D) = n.
\]
If $\rank(D) = n-1$, then $\operatorname{nullity}(D) = 1$, so $\ker(D)$ is the one-dimensional space spanned by $u$; consequently the unique nonzero solution to $D u' = 0$ is $u$. Gaussian elimination finds it in time polynomial in $n$ and $\log(1/\epsilon)$. It remains to show that $\mathbb{P}(\rank(D) = n-1) \geq 1-\epsilon$.
If the rows fail to span $V$, then all $d_j$ lie in some proper $(n-2)$-dimensional subspace $W \subset V$.
For any fixed $W$, we have
\[
\mathbb{P}(d_j \in W) = \frac{|W|}{|V|} = \frac{2^{n-2}}{2^{n-1}} = \tfrac{1}{2},
\]
independently across $j$, hence $\mathbb{P}(\forall j,\ d_j \in W) = 2^{-m+1}$.
The number of $(n-2)$-dimensional subspaces of a $(n-1)$-dimensional $\mathbb{F}_2$-space is given by the Gaussian binomial
\[
\binom{n-1}{n-2}_2 = \binom{n-1}{1}_2 = \frac{2^{n-1}-1}{2-1} = 2^{n-1} - 1.
\]
By a union bound,
\[
\mathbb{P}(\rank(D) < n-1 )
\leq (2^{n-1} - 1) \, 2^{-m+1}
< 2^{n-m}
\leq \epsilon,
\]
for our choice of $m$.
\end{proof}

\paragraph{Lower bound.}
Next, we will show that  the bound is tight in $n$, i.e., $\Omega(n)$ samples are necessary. We follow a similar proof as in~\cite{benedetti2025complement}, and we find it convenient to talk about \emph{success probability} $p$ instead of expected score $V$, but these are related as $V = 2p -1$.

\begin{theorem}
Fix $n\geq 1$. Any classical algorithm that, given $q$ samples from $S\in\mathcal{F}_{\mathrm{BV}}$, outputs some $y\in\bar S$ has success probability at most
\[
\frac{1}{2} + \frac{1}{2\bigl(2^{ n-q+1}-1\bigr)}\qquad\text{for }1 \leq q \leq n.
\]
\end{theorem}

\begin{proof}
Just as in~\cite{benedetti2025complement}, we assume query access to $S$ by an oracle which on input $i$ returns the $i$-th element of $S$, assuming arbitrary (and unknown) ordering of its elements. By Yao’s minimax principle and since the index-query model only strengthens classical access relative to sampling~\cite{benedetti2025complement}, it suffices to bound the success of any deterministic index-query algorithm against the hard distribution where $(u,b)$ is uniform over $\Set{0,1}^n\setminus\Set{0}^n \times\Set{0, 1}$ and $S=S_{u,b}$.
Let $x_1,\dots,x_q$ be any $q$ distinct elements returned by the oracle. Let $W=\operatorname{span}\Set{x_1,\dots,x_q} $ with $r=\dim W \leq q$. The constraints $x_i\in S$ mean that $u\cdot x_i=b$ for all $i$. The set of feasible pairs $(u,b)$ consistent with the observed $x_1,\dots,x_q$ is a linear subspace of $\mathbb{F}_2^{n}\times\mathbb{F}_2$ of dimension
\[
d = n-r+1,
\]
since we have an $(n+1)$-dimensional space with $r$ independent linear constraints. Under the uniform prior on $(u,b)$ with $u\neq \Set{0}^n$, conditioning on $S$ containing $x_1,\dots,x_q$ yields a uniform posterior over the remaining feasible pairs. This follows from a similar argument as in~\cite{benedetti2025complement}: for any fixed subset $T \subseteq \Set{0,1}^n$ of rank $r$, the number of pairs $(u,b)$ with $T \subseteq S$ is the same. Their number is
\[
N_{\mathrm{rem}}  = 2^{ d}-1  = 2^{ n-r+1}-1.
\]
Equivalently, the feasible $u$’s form a $d$-dimensional linear subspace in $\mathbb{F}_2^n$. Write $U$ for the associated $d$-dimensional linear subspace. Conditioned on the constraints being satisfied, $u$ is uniform over $U\setminus\Set{0}^n$ and $b$ is fixed by $b \equiv u\cdot x_i \pmod{2}$. Let the algorithm output any $y\notin\Set{x_1,\dots,x_q}$ such that there exists an $u'\in U\setminus\Set{0}^n$ with $u'\cdot y \equiv 1 - b(u') \pmod{2} $, where $b(u') \equiv u'\cdot x_1  \pmod{2}$. We may assume without loss of generality that the algorithm chooses such a $y$, since any $y$ with $u'\cdot y \equiv b(u') \pmod{2}$ for all $u'\in U$ would make the success probability zero. For such a choice of $y$, the map $\phi_y: U\to\Set{0,1}$ defined as
\[
\phi_y(u') = u'\cdot(y+x_1) \pmod{2} ,
\]
is a nonzero linear functional on the $d$-dimensional space $U$. Note that $\phi_y(u') = u' \cdot y + b(u') \equiv 1 \pmod{2}$ if and only if $u' \cdot y \equiv 1- b(u') \pmod{2}$, which means that $\phi_y(u') = 1$ if and only if the algorithm succeeds for $(u',b(u'))$. Note that $\phi_y(u')$ takes each value in $\Set{0,1}$ exactly $2^{d-1}$ times as $u'$ ranges over $U$. Excluding $u'=\Set{0}^n$ removes at most one element and never changes the count for value $1$, so among the $N_{\mathrm{rem}}=2^d-1$ feasible inputs, exactly $2^{d-1}$ yield $\phi_y(u') = 1$. Therefore
\[
\mathbb{P}(\text{success}\mid x_1,\dots,x_q)\le\frac{2^{d-1}}{2^d-1},
\]
where the success probability is with respect to the distribution over inputs. Since this is monotonically decreasing in $d$, this is minimized when $r=q$, i.e., the queried elements are linearly independent, giving $d=n-q+1$. Since this holds for the constraints given any elements $x_1,\dots,x_q$ and for any output $y$ as specified above, the following final bound on the success probability holds:
\begin{align}
\frac{2^{ n-q}}{2^{ n-q+1}-1} = \frac{1}{2}+\frac{1}{2\bigl(2^{ n-q+1}-1\bigr)}.
\label{eq:success_prob_lb_BV_states}
\end{align}
\end{proof}
Therefore, to achieve \emph{constant} success probability $p > \tfrac{1}{2}$, we require $q = \Omega(n)$. Note that for $q=1$, we recover Eq.~\eqref{eq:quantum_score_BV_states} when we map Eq.~\eqref{eq:success_prob_lb_BV_states} to the expected score via $V = 2p -1$.

\section{The rigidity of complement sampling}
\label{app:verify_state_prep}

In this Supplementary Note we prove results for complement sampling which allows us to play the game without a trusted state preparation, and trusted strategy instead. We also prove that if both parties use untrusted devices then they could in principle collude to achieve perfect score. 
Recall the definition of rigidity score, $R(\mathcal{F}, h) = \mathbb{E}_S \mathbb{E}_c  \mathbb{E}_{y|S,c}[(-1)^c \, \sigma(S, y)]$. 
We begin with a lower bound on the expected fidelity of state preparation when the player adopts the optimal quantum strategy $Q$.  

\begin{theorem}
\label{thm:subset_test_forces_closeness}
Let $n$ be a positive integer, let $\epsilon \in [0,1]$, and let $\mathcal{F} = \Set{S_i \given S_i \subseteq \Set{0,1}^n, \abs{S_i} = 2^{n-1}}$ Suppose the referee has access to an untrusted quantum device that on input $S \in \mathcal{F}$ prepares an unknown state $\rho_{SE}$. The referee samples $S \in \mathcal{F}$ and $c\in\Set{0,1}$ uniformly at random: if $c=0$, $\rho_{S} = \tr_E(\rho_{SE})$ is sent to the player, who uses the optimal quantum strategy to return a $n$-bit string $y$; if $c=1$ instead, the referee measures $\rho_{S}$ in the computational basis to obtain a $n$-bit string $y$. Then
\[
R(\mathcal{F}, Q) \geq 1 - \epsilon \; \implies \; \mathbb{E}_S[ \bra{S} \rho_S \ket{S}] \geq 1 - \epsilon .
\]
\end{theorem}

\begin{proof}
Fix the subset $S$ and define the projector $\Pi_S\coloneqq \sum_{x\in S}\dyad{x}$. Then 
\begin{align*}
    \mathbb{E}_c  \mathbb{E}_{y|S,c}[(-1)^c \, \sigma(S, y)] &= \frac{1}{2} \mathbb{E}_{y|S,c=0}[ \sigma(S, y)] + \frac{1}{2} \mathbb{E}_{y|S,c=1}[-\sigma(S, y)] \\
    &= \frac{1}{2}\left(1 - 2\tr(\rho_S U \Pi_S U^\dag)\right) + \frac{1}{2}\left(2\tr(\rho_S \Pi_S) -1\right)\\
    &= \tr(\rho_S (\Pi_S - U \Pi_S U^\dag)) ,
\end{align*}
where we substitute the scores for when the referee measures directly and for when the player plays the optimal quantum strategy. Let us define the projector $\Pi_+\coloneqq \dyad{+^n}$ so that under the optimal quantum strategy we have $U = 2\Pi_+ - I$. It is easy to verify the following properties:
\[
\Pi_+\Pi_S=\frac{1}{\sqrt{2}}\dyad{+^n}{S},
\qquad
\Pi_S \Pi_+=\frac{1}{\sqrt{2}}\dyad{S}{+^n},
\qquad
\Pi_+\Pi_S \Pi_+=\frac{1}{2}\dyad{+^n}.
\] 
Thus
\begin{align*}
\Pi_S - U \Pi_S U^\dag &= \Pi_S - (2\Pi_+ - I) \Pi_S (2\Pi_+ - I) \\
&= -4 \Pi_+ \Pi_S \Pi_+ +2\Pi_S \Pi_+ + 2\Pi_+ \Pi_S \\
&= -2 \dyad{+^n} + \sqrt{2} \ket{S} \bra{+^n} + \sqrt{2} \ket{S} \bra{+^n} \\
&= - \left( \sqrt{2} \ket{+^n} - \ket{S}\right)\left( \sqrt{2} \bra{+^n} - \bra{S}\right)  + \dyad{S} \\
&= - \dyad{\bar{S}} + \dyad{S} ,
\end{align*}
where we complete the square and use that $\ket{\bar{S}} =  \sqrt{2} \ket{+^n} - \ket{S}$. 
Plugging this result back into the previous step we have
\begin{equation}\label{eq:score_fidelity_inequality}
    \mathbb{E}_c  \mathbb{E}_{y|S,c}[(-1)^c \, \sigma(S, y)] = \tr(\rho_S \dyad{S}) - \tr(\rho_S \dyad{\bar{S}}) \leq  \bra{S} \rho_S \ket{S} .
\end{equation}
Taking the expectation over subsets we have
\[
    \mathbb{E}_S \mathbb{E}_c  \mathbb{E}_{y|S,c}[(-1)^c \, \sigma(S, y)] \leq \mathbb{E}_S[ \bra{S} \rho_S \ket{S}], 
\]
and thus
\[
    R(\mathcal{F}, Q) \geq 1 - \epsilon \; \implies \; \mathbb{E}_S[ \bra{S} \rho_S \ket{S}] \geq 1 - \epsilon .
\]
\end{proof}

Theorem~\ref{thm:subset_test_forces_closeness} can be used to certify a lower bound on the fidelity of the referee’s state preparation, assuming a trusted player device. 
Let us also look at the converse statement.

\begin{proposition}
\label{prop:closeness_forces_score}
In the setting of Theorem~\ref{thm:subset_test_forces_closeness}, we also have that
\[
\mathbb{E}_S[ \bra{S} \rho_S \ket{S}] \geq 1 - \epsilon  \implies R(\mathcal{F}, Q) \geq 1 - 2\epsilon .
\]
\end{proposition}

\begin{proof}
Note that $\dyad{S} + \dyad{\bar{S}} \preceq I$ because the two projectors are rank-1 and onto orthogonal subspaces. Taking the expectation of both sides and rearranging we have $\tr(\rho_S \dyad{\bar{S}}) \leq 1 - \tr(\rho_S \dyad{S})$. Plugging this into the equality in Eq.~\eqref{eq:score_fidelity_inequality} we have
\[
\mathbb{E}_c  \mathbb{E}_{y|S,c}[(-1)^c \, \sigma(S, y)] = \tr(\rho_S \dyad{S}) - \tr(\rho_S \dyad{\bar{S}})  \geq 2 \tr(\rho_S \dyad{S}) -1. 
\]
Taking the expectation over subsets we have
\[
\mathbb{E}_S \mathbb{E}_c  \mathbb{E}_{y|S,c}[(-1)^c \, \sigma(S, y)] \geq 2 \mathbb{E}_S[\bra{S} \rho_S \ket{S}] - 1,
\]
and it follows that
\[
\mathbb{E}_S[ \bra{S} \rho_S \ket{S}] \geq 1 - \epsilon  \implies R(\mathcal{F}, Q) \geq 1 - 2\epsilon .
\]
\end{proof}

Next, we upper-bound the rigidity score under the assumption that the player uses a specific trusted optimal classical strategy $C$. Specifically, upon receiving a string $x$, the player samples uniformly from $\Set{0,1}^n\setminus\Set{x}$, which is optimal in the noise-free setting because it is fixed-point-free (see \ref{app:proof_max_violation}).

\begin{proposition}
\label{prop:classical_BV_rigidity_bound}
Let $n$ be a positive integer, and let $\mathcal{F}_{\mathrm{BV}}=\Set{S_{u,b}\given u\in\Set{0,1}^n\setminus\Set{0^n}, b\in\Set{0,1}}$. Suppose the referee has access to an untrusted device that on input $S\in\mathcal{F}_{\mathrm{BV}}$ outputs a random $x \in \Set{ 0,1}^n$. Next, it samples $c\in\Set{0,1}$ uniformly at random: if $c=0$, $x$ is sent to the player, who samples a uniform $y \in \Set{0,1}^n\setminus\Set{x}$ and returns it; if $c=1$ instead, the referee checks if $x\in S$. Then
\[
    R(\mathcal{F}_{\mathrm{BV}},C)
    \leq
    \frac12+\frac{1}{2(2^n-1)} .
\]
\end{proposition}

\begin{proof}
We first consider the $c=0$ branch. Fix $S_{u,b}\in\mathcal{F}_{\mathrm{BV}}$ and a string $x \in \Set{0,1}^n$ received by the player. If $x \in S_{u,b}$, we have
\[
    \mathbb{P}(y\notin S_{u,b}\mid S_{u,b},x \in S_{u,b})
    =
    \frac{2^{n-1}}{2^n-1},
\]
whereas, if $x \notin S_{u,b}$,
\[
    \mathbb{P}(y\notin S_{u,b}\mid S_{u,b},x \in S_{u,b})
    =
    \frac{2^{n-1}-1}{2^n-1}.
\]
Using $\sigma(S,y)=2\iv{y\notin S}-1$, it follows that
\[
    \mathbb{E}_{y|S_{u,b}}[\sigma(S_{u,b},y)]
    =
    \begin{cases}
        \frac{1}{2^n-1}, & x\in S_{u,b},\\[1mm]
        -\frac{1}{2^n-1}, & x \notin S_{u,b}.
    \end{cases}
\]
In particular, for every $S_{u,b}$ and every possible received string $x$, this conditional expectation is at most $1/(2^n-1)$ (which is the $x \notin S_{u,b}$ case). Hence, the expected score in the $c=0$ branch is at most $1/(2^n-1)$. In the $c=1$ branch, the referee checks whether $x\in S_{u,b}$, which succeeds with probability at most $1$. Since $c$ is uniform, we obtain
\[
    R(\mathcal{F}_{\mathrm{BV}},C)
    \leq
    \frac{1}{2}\cdot \frac{1}{2^n-1}
    +
    \frac{1}{2}\cdot 1
    =
    \frac{1}{2}+\frac{1}{2(2^n-1)} .
\]
\end{proof}

Finally, we now briefly explain why the assumption that the classical strategy is trusted is necessary when the device is untrusted. If both the state preparation and the player's device are untrusted, then they can collude in a way that passes the relevant checks while producing samples that are far from uniform over the intended subset.

\begin{proposition}
\label{prop:colluding_classical_sampler_player} 
Let $n$ be a positive integer. There exists a classical colluding strategy for an untrusted state preparation and player's device such that
\[
    R(\mathcal{F}_{\mathrm{BV}},C)  = 1.
\]
In this colluding strategy, for every $S_{u,b}=\Set{x\in\mathbb{F}_2^n \given u\cdot x=b}$, the state preparation and the player only consider a subset of $S_{u,b}$ of size $2^{n-2}$.
\end{proposition}
\begin{proof}
The case of $n = 1$ is trivial since the state preparation can send the only $x$ in the subset and the player can output $y = x \oplus 1$. Let us now consider $n \geq 2$. Choose a matrix $A$ such that both $A$ and $I-A$ are invertible. Such a matrix
exists for every $n\geq 2$: identify $\mathbb{F}_2^n$ with the field $\mathbb{F}_{2^n}$ and choose $\alpha\in\mathbb{F}_{2^n}\setminus\Set{0,1}$, which is possible since $n\geq 2$. Let $A$ be the $\mathbb{F}_2$-linear map given by multiplication by $\alpha$. Then $A$ is invertible because $\alpha\neq 0$, and $I-A$ is invertible because it is multiplication by $1-\alpha\neq 0$. For each $S_{u,b}$, define
\[
    G_{u,b}
    =
    \Set{x\in\mathbb{F}_2^n \given u\cdot x=b,\ u\cdot Ax=b\oplus 1}.
\]
We claim that $|G_{u,b}|=2^{n-2}$. Indeed, the two constraints are linear constraints on $\mathbb{F}_2^n$. They are independent because the corresponding linear forms are $u$ and $A^T u$, and $A^T u\neq u$: otherwise $(I-A^T)u=0$, contradicting the invertibility of $I-A$. Hence the solution space has size $2^{n-2}$. The colluding strategy is now simple. On input $S_{u,b}$, the state preparation produces $x$ uniformly from $G_{u,b}$ and sends it to the player. By definition, $x\in S_{u,b}$. The player, upon receiving $x$, outputs $y=Ax$. Again by definition of $G_{u,b}$, we have $u\cdot Ax=b\oplus 1$, and therefore $y\notin S_{u,b}$. Thus the state preparation always passes the membership check and the player always outputs an element of the complement, which implies that $R(\mathcal{F}_{\mathrm{BV}},C) = 1$.
\end{proof}

The only way to detect the above attack from samples is to distinguish the uniform distribution on $S_{u,b}$, which has support size $2^{n-1}$, from the uniform distribution on the hidden subset $G_{u,b}\subseteq S_{u,b}$, which has support size $2^{n-2}$. This requires exponentially many samples in the black-box sampling model. Hence a polynomially bounded referee cannot reliably rule out such a colluding strategy using statistics alone.

\section{Analysis of the game with pseudorandom permutations subsets}
\label{app:details_prp_analysis}

As discussed in the main text, a subset state can be generated by applying a permutation $P$ to $\ket{0} \otimes \ket{+^{n-1}}$. However, a reversible circuit implementation has exponential depth for general $P$. The original complement sampling game~\cite{benedetti2025complement} proposes to use strong pseudorandom permutations (PRPs) instead. Strong PRPs typically have polynomial circuit depth and are indistinguishable from truly random permutations in polynomial time, even when given access to their inverse. The following result shows that complement sampling with these subsets remains hard in the average case for any classical strategy.

\begin{theorem}[restatement of Thm. 7 in Ref.~\cite{benedetti2025complement}]
    For some $n \in \mathbb{N}$, let $P : \Set{0,1}^n \times \Set{0,1}^{\lambda} \rightarrow \Set{0,1}^n$ be a strong pseudorandom permutation with security parameter $\lambda = n$. Given a key $k \in \Set{0,1}^{\lambda}$, let $S$ be the image of $\Set{0ik \given i \in \Set{0,1}^{n-1} }$ under $P$, and $\bar{S}$ be the image of $\Set{1ik \given i\in \Set{0,1}^{n-1} }$ under $P$. Let $O_{\textup{index}} : [2^{n-1}] \rightarrow \Set{0,1}^{n}$ be the oracle that on input $i$ returns the $i$-th element in $S$.  Picking a key uniformly at random, for all polynomial-time algorithms $A$ that make a polynomial number of queries to $O_{\textup{index}}$ it must hold that
    \begin{align*}
        \mathbb{P} \left(A^{O_{\textup{index}}} \textup{ outputs a } y \in \bar{S}\right) \leq \frac{1}{2} + \textup{negl}(n).
    \end{align*}
\label{thm:hardness_PRP_S}
\end{theorem}

Luby and Rackoff~\cite{luby1988how} showed that strong PRPs can be constructed from PRPs. PRPs used in cryptographic protocols such as Advanced Encryption Standard (AES) and Data Encryption Standard (DES) could be used for our purpose. However, they are made of rigid circuit constructions (the block ciphers) operating on a fixed number of bits (the block size), which would complicate the scaling analysis. To explore more flexible constructions we instead consider ensembles of $\epsilon$-approximate $k$-wise independent PRPs. When applied to any fixed initial bit string, these PRPs produce distributions of bit strings that are indistinguishable from the uniform distribution up to the $k$-th moment and up to total variation distance $\epsilon$. 

\begin{definition}[restatement of Def.~1 in Ref.~\cite{gay2025pseudorandomness}]
\label{def:kwise_prp}
A distribution $P$ on permutations of $\Set{0, 1}^n$ is said to be $\epsilon$-approximate $k$-wise independent if for all distinct $x(1),\dots, x(k) \in \Set{0, 1}^n$, the distribution of $(g(x(1)),\dots, g(x(k)))$ for
$g \sim P$ has total variation distance at most $\epsilon$ from the uniform distribution on $k$-tuples of distinct strings from $\Set{0, 1}^n$
\end{definition}

In the same paper, the authors provide circuit constructions based on layers of random reversible operations.

\begin{theorem}[restatement of Thm.~3 in Ref.~\cite{gay2025pseudorandomness}]
\label{thm:prp_theorem}
A random one-dimensional brickwork circuit of three-bit gates and depth $n \cdot \tilde{\mathcal{O}}(k^2)$ computes a permutation that is $2^{-nk}$-approximate $k$-wise independent.
\end{theorem}

More generally, the idea is to build a reversible circuit from layers of random three-bit operations, so that inputs and outputs become increasingly uncorrelated as the number of layers grows. To explore variations of this construction, we define three circuit parameters: layer type, operation type, and number of layers.
The \emph{layer type} prescribes the location of the operations. We consider either placing them in a one-dimensional nearest-neighbors brickwork (NN), or as a layout of random triplet (RT). More specifically, NN is a sequence of $\frac{n}{3}$ three-bits operations starting from the first qubit $(1,2,3), (4,5,6), \dots, (n-2,n-1,n)$, followed by another sequence starting from the second qubit $(2,3,4), (5,6,7), \dots, (n-1,n,1)$. In RT we proceed similarly except that each triplet is selected at random. For example, we could have $(1, 6, 8), (5, 2, 9), (3, 4, 7)$, followed by $(3, 7, 8), (1, 5, 6), (9, 2, 4)$. The RT construction is motivated by the full-connectivity of Quantinuum H-series devices and is similar in spirit to the construction used for random circuit sampling in Ref.~\cite{DeCross_2025}. The \emph{operation type} can be any gate set which is universal for reversible classical computation. We consider DES2 gates~\cite{Even_1983}, which are also used in the proof of Theorem~\ref{thm:prp_theorem}. DES2 gates act on the three bits as $(b_1, b_2, b_3) \rightarrow g(b_1, b_2, b_3 \oplus f(b_1,b_2))$ for any of the $2^{2^2} = 16$ possible binary functions $f$, and any of the $3! = 6$ possible reordering $g$ of the output bits. Thus there are $96$ possible instantiations of a DES2 gate on three bits. In addition, we consider S8 gates that act as arbitrary permutations of the eight basis states of three bits. S8 is a superset of DES2 and there exists $8! = 40320$ possible instantiations on three bits. Finally, we consider either a constant \emph{number of layers} $L=1$, or a number of layers that scales with the number of bits $L=n$. We will refer to our PRP constructions using the name convention: layer type - operation type - number of layers. For example, NN-DES2-$n$ consists of $n$ layers of nearest-neighbors DES2 gates. 

\begin{table}
\centering
\small
\begin{tabular}{c|c|c|c|c}
Layer type  & Operation type & Number of layers & Average 2-qubit gates & Worst-case 2-qubit gates \\
\hline
NN & DES2 & $1$ & $2n$      & $10n/3$ \\
NN & DES2 & $n$ & $2n^2$    & $10n^2/3$ \\
NN & S8   & $1$ & $20n/3$   & $12 n$ \\
NN & S8   & $n$ & $20n^2/3$ & $12 n^2$ \\
RT & DES2 & $1$ & $2n$      & $10n/3$ \\
RT & DES2 & $n$ & $2n^2$    & $10n^2/3$ \\
RT & S8   & $1$ & $20n/3$   & $12 n$ \\
RT & S8   & $n$ & $20n^2/3$ & $12 n^2$ \\
\end{tabular}
\caption{\textit{Pseudorandom permutations constructions.}}
\label{tab:prp_constructions}
\end{table}

Supplementary Table~\ref{tab:prp_constructions} shows the constructions considered in this study and reports the average- and worst-case scaling of number of two-qubit gates. These are estimated by recalling that each layer contains $2n/3$ operations and by decomposing each operation to quantum gates as explained in~\ref{app:details_circ_compilation}. 
The estimates are validated in Supplementary Fig.~\ref{fig:prp_analysis}(b) where we take $n$ to be a multiple of 3 and report the average number of two-qubit gates of 100 random instances. Each circuit is compiled for the Quantinuum H-series native gate set $R_{ZZ}(\theta) = e^{-i \frac{\theta}{2} Z\otimes Z}$ and $U(\theta, \phi) = e^{-i\frac{\theta}{2} (\cos(\phi)X + \sin(\phi)Y)}$ using optimization level 3 of the \texttt{pytket} software package~\cite{Sivarajah_2021}. The number of gates increases linearly for single layer constructions and quadratically for $n$ layer constructions, as expected. We note a significant overhead for implementing S8 (blue lines) compared to DES2 gates (red lines), while the layer type did not have a significant impact.

Next we numerically verify that our PRP constructions are $\epsilon$-approximate $k$-wise independent where, following Theorem~\ref{thm:prp_theorem}, we choose $\epsilon = 2^{-kn}$. The idea is to sample random instances, apply them to the all-zero bit string, and measure in the computational basis. We then compare the distribution of output bit strings $p_i$ to the uniform distribution in terms of total variation distance $\mathrm{TVD} = \frac{1}{2}\sum_{i=1}^{2^n} \abs{p_i - \frac{1}{2^n}}$. A sample-optimal algorithm for this task is given in Ref.~\cite{diakonikolas2017optimal}. It computes the random variable $S = \frac{1}{2} \sum_{i=1}^{2^n} \abs{ \tilde{p}_i - \frac{1}{2^n} }$ where $\tilde{p}_i$ are empirical probabilities estimated from the $M$ samples. If $S$ is larger than a specific threshold, the algorithm outputs `NO', otherwise it returns `YES'. The authors prove the following theorem.

\begin{theorem}[restatement of Thm. 3 in Ref.~\cite{diakonikolas2017optimal}]
\label{thm:tvd_samples}
Let $p$ be an unknown distribution over $2^n$ elements. There exists a universal constant $C>0$ such that the following holds. Given $M 
\geq \frac{C}{\epsilon^2}(\sqrt{2^n \log \frac{1}{\delta}} + \log \frac{1}{\delta})$ samples from $p$, there exists an algorithm that outputs `YES' with probability $1-\delta$ if $p$ is the uniform distribution, and `NO' with probability $1-\delta$ if the total variation distance to the uniform distribution is $\geq \epsilon$. 
\end{theorem}

Using this theorem with $\epsilon = 2^{-kn}$ and $k=1$, we observe that $\mathcal{O}(2^{2.5n})$ samples are sufficient for the algorithm in~\cite{diakonikolas2017optimal} to solve the decision problem. In our work we use $S$ as an estimate for the TVD, so we don't threshold its value. For each PRP construction in Table~\ref{tab:prp_constructions} we sample $10 n \times 2^{2.5}$ random instances, where the factor of $10n$ is chosen as the maximum we can achieve with our limited computational capabilities. Note that we could not analyze the second moment, $k=2$, because it requires the concatenation of 2 input bit strings (see Def.~\ref{def:kwise_prp}) leading to a number of samples scaling as $\mathcal{O}(2^{5n})$.

\begin{figure}
    \centering    
    \includegraphics[width=\linewidth]{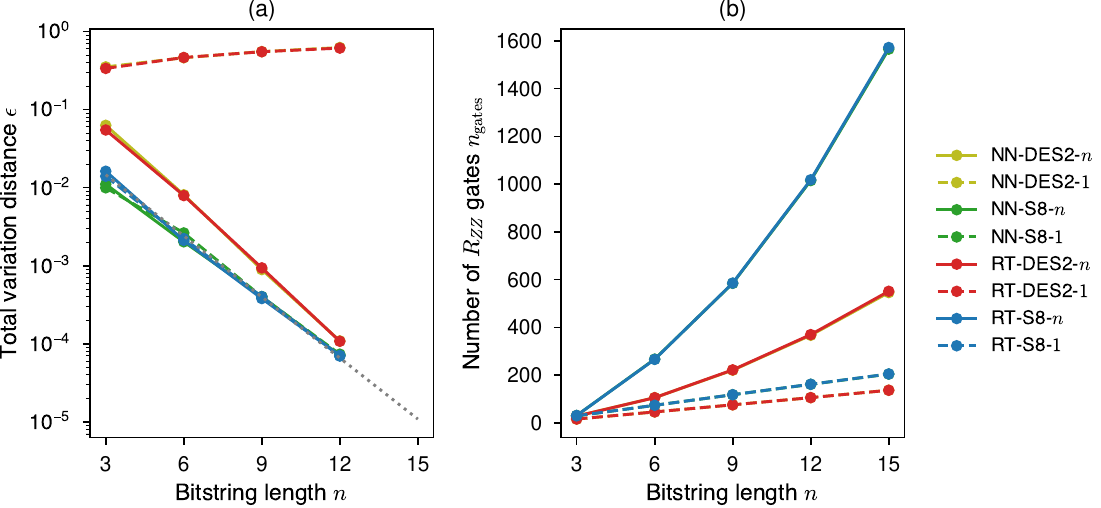} 
    \caption{\textit{Numerical analysis of pseudorandom permutations constructions.} (a) Total variation distance $\epsilon$ to the uniform distribution versus number of bits $n$. (b) Number of two-qubit $R_{ZZ}$ gates $n_\mathrm{gates}$ required for a reversible circuit implementation optimized for Quantinuum H2-1. Each point is averaged over $100$ random instances. Data corresponding to NN constructions lie directly below the corresponding RT constructions.}
    \label{fig:prp_analysis}
\end{figure}

Supplementary Fig.~\ref{fig:prp_analysis}(a) shows that $\epsilon$ decreases exponentially in $n$ for all constructions except NN-DES2-1 and RT-DES2-1. This indicates that a single layer of DES2 gates does not result in sufficient information scrambling, while a single layer of S8 gates does. With regard to the connectivity of the layers, we do not observe a significant difference between nearest-neighbors (NN) and random triplets (RT). Finally, from~\cite[Theorem 3]{gay2025pseudorandomness} it follows that NN-DES2-$n$ asymptotically provides $2^{-n}$-approximate 1-wise independent pseudorandom permutations, which is confirmed here numerically (yellow solid line). We extrapolate to $n = 15$ (dotted gray line) for RT-S8-1 and $n=36$ (not shown) which are setups used later on in hardware experiments. Finally, we shall point out a subtlety: since S8 gates are \emph{exact} permutations on three bits, we would expect the TVD for $n=3$ (green and blue lines) to be zero. This does not happen because we use a finite number of samples, albeit sufficient for our purposes due to Theorem~\ref{thm:tvd_samples}. We confirm that when using ten times more samples the TVD of S8 decreases, and that of DES2 remains the same. Since our analysis concerns upper bounds to the TVD, this fact does not affect our conclusions. 

We use the data in Supplementary Fig.~\ref{fig:prp_analysis} to choose suitable settings for executing the complement sampling game with PRP subsets on Quantinuum H-series hardware. To this end, we shall take into account also the linear cost of two-qubit gates for implementing both teleportation and swapper circuits (see~\ref{app:details_circ_compilation}). 
We choose to execute RT-DES2-$n$ with $n=12$, yielding $n_\mathrm{qubits}=41$ and an average of $n_\mathrm{gates} = 370$, as well as RT-S8-1 with $n=15$, yielding $n_\mathrm{qubits}=51$ and an average of $n_\mathrm{gates} = 205$. We also execute the game without the teleportation protocol, and with the referee and player sharing the same register. This allows us to explore much longer bit strings. We consider RT-DES2-$n$ with $n=36$, yielding $n_\mathrm{qubits}=53$ and an average of $n_\mathrm{gates} = 438$.

We compare the empirical score $\hat{V}$ of the implemented quantum strategy against the score of an optimal classical strategy that can somehow exploit the approximation error in the PRP family. Denoting the family of approximate PRPs as $\mathcal{F}_{\mathrm{PRP},\epsilon}$, the optimal classical strategy achieves $V(\mathcal{F}_{\mathrm{PRP},\epsilon}, C) = \frac{1}{2^n - 1} + 2\epsilon$. Here $\epsilon$ is found by extrapolating the data in Supplementary Fig.~\ref{fig:prp_analysis} and a factor of $2$ arises from the conversion from TVD to score. Each game is executed for 800 rounds on the H2-1 and H2-2 quantum computers.
This amounts to running 4800 quantum circuits and obtaining a single sample from each. Our setup is summarized in Supplementary Table~\ref{tab:experiment_setup_prp}. At the time of these experiments in July 2025, H2-1 had two-qubit gate infidelity of $(1.05 \pm 0.08) \times 10^{-3}$ and memory error per qubit and depth-1 circuit time of $(2.03 \pm 0.23) \times 10^{-4}$. H2-2 had two-qubit infidelity of $(1.10 \pm 0.17)\times 10^{-3}$ and memory error per qubit and depth-1 circuit time of $(2.26 \pm 0.29) \times 10^{-4}$. These were measured via randomized benchmarking of all qubits~\cite{Moses_2023,Magesan_2011,Harper_2019,hseries_benchmarking}.

\begin{table}
\centering
\small
\begin{tabular}{c|c|c|c|c|c|c} 
$n$ & PRP & $\epsilon$ & teleportation & $n_\mathrm{qubits}$ & $n_\mathrm{gates}$ avg & $n_\mathrm{gates}$ std \\ 
\hline 
12  & RT-DES2-$n$ & $1.078 \times 10^{-4}$  & yes & 41 & 370.47 & 19.96 \\ 
15  & RT-S8-1     & $1.087 \times 10^{-5}$  & yes & 51 & 205.16 & 9.45 \\ 
36  & RT-S8-1     & $3.639 \times 10^{-11}$ & no  & 53  & 438.10 & 14.72 \\
\end{tabular}
\caption{\textit{Summary of the experimental settings for PRP subset states.}} 
\label{tab:experiment_setup_prp} 
\end{table}

\begin{figure}
\centering
\includegraphics[width=\linewidth]{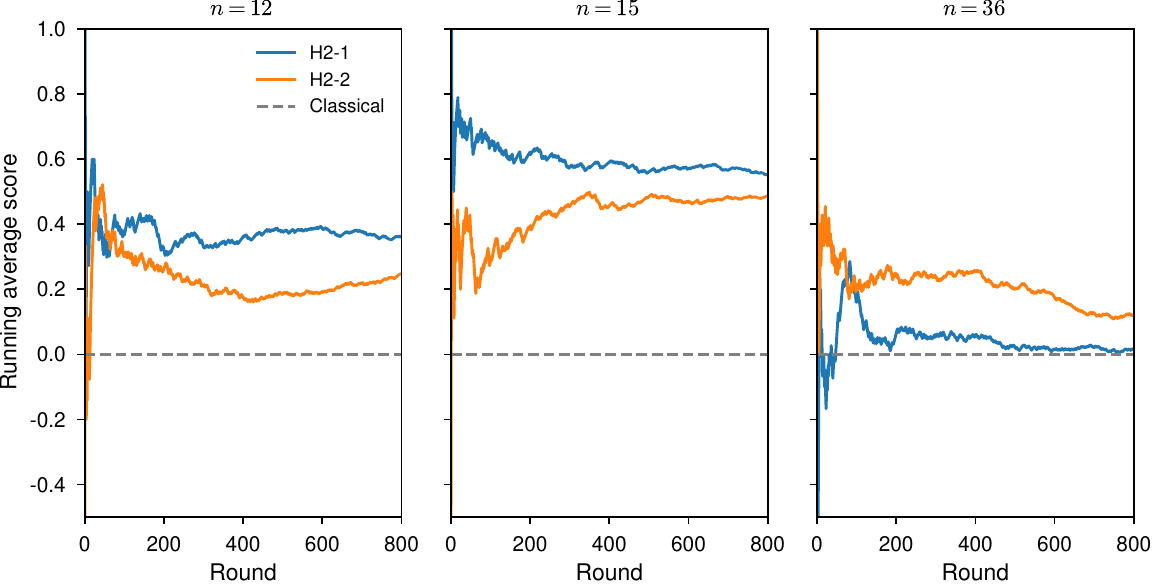}
\caption{\textit{Experimental results for PRP subset states.} Running average of the player's score plotted against the number of rounds. The optimal classical strategy achieves the score indicated by the gray line.}
\label{fig:experiment_results_prp}
\end{figure}

Supplementary Fig.~\ref{fig:experiment_results_prp} shows the running average of the empirical score during the game. In almost all settings the quantum strategy outperforms the optimal classical strategy (dashed gray line) significantly. This provides evidence of a large violation of classicality, albeit conditional on the existence of PRPs or, equivalently, on the existence of one-way functions~\cite{katz2020introduction}. Noting the low performance on the $n=36$ run with H2-1, we choose to employ H2-2 in the experiments presented in the main text.

\section{Details about the implementation}
\label{app:details_circ_compilation}

\paragraph{Quantum teleportation.} A quantum communication channel is used in some experiments with both BV and PRP subset states. A simple way to simulate a quantum communication channel between two parties is to swap their qubits. This `symmetric' communication channel implies that while the referee sends a state to the player, the player simultaneously sends back a state to the referee. We opted for a more realistic and resource consuming `asymmetric' quantum communication channel. This is provided by the well-known quantum teleportation protocol~\cite{Bennett_1993}. The two parties share a set of Bell pairs, the referee performs local operations and measurements on them, and then communicates the outcomes via a classical channel. The player applies local operations on the Bell pairs accordingly, effectively recovering the referee's state. 

In the main text, Fig.~1C shows the textbook implementation of quantum teleportation for a single qubit in the state $\ket{S}$. From Supplementary Fig.~\ref{fig:h2_gateset_dec}, a CX can be decomposed into a $R_{ZZ}$ gate up to a few single-qubit gates. The mid-circuit measurements are followed by classical communication and correcting operations. These are implemented on the Quantinuum H2 hardware using conditional gates. For the implementation of an $n$-qubit state we simply apply this circuit to each qubit and associated ancilla.

In summary, for a state of $n$ qubits the communication between referee and player requires $n$ ancilla qubits, $2n$ two-qubit gates, $2n$ mid-circuit measurements, and $2n$ conditional single-qubit gates.

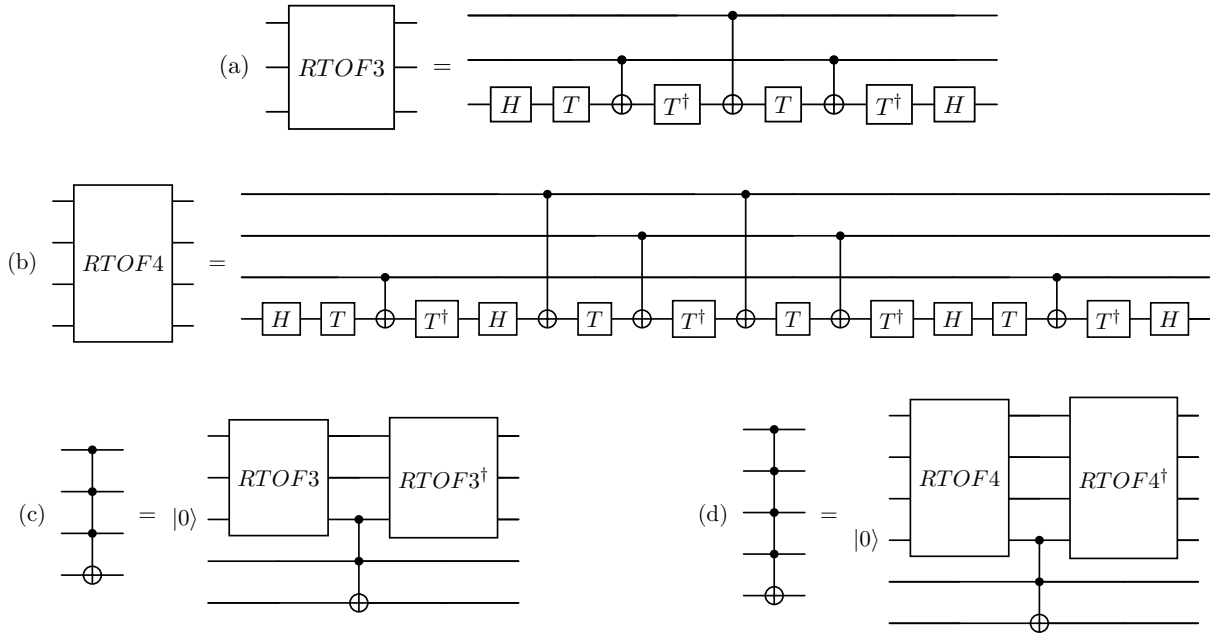
\begin{figure}
\centering
\begin{adjustbox}{width=.65\textwidth}
(a)
    \begin{quantikz}[column sep=10pt, row sep={20pt,between origins}]
    & \gate[3]{RTOF3} &  \\
    & & \\
    & & 
    \end{quantikz}
    =
    \begin{quantikz}[column sep=10pt, row sep={20pt,between origins}]
    & \qw      & \qw      & \qw      & \qw           & \ctrl{2} & \qw      & \qw      & \qw           & \qw      &  \\
    & \qw      & \qw      & \ctrl{1} & \qw           & \qw      & \qw      & \ctrl{1} & \qw           & \qw      &  \\
    & \gate{H} & \gate{T} & \targ{}  & \gate{T^\dag} & \targ{}  & \gate{T} & \targ{}  & \gate{T^\dag} & \gate{H} & 
    \end{quantikz}
\end{adjustbox}
\\ 
\vspace{.5cm}
\begin{adjustbox}{width=1\textwidth}
(b)
    \begin{quantikz}[column sep=10pt, row sep={20pt,between origins}]
    & \gate[4]{RTOF4} &  \\
    & & \\
    & & \\
    & & 
    \end{quantikz}
    =
    \begin{quantikz}[column sep=10pt, row sep={20pt,between origins}]
    & \qw   & \qw  & \qw  & \qw     & \qw   & \ctrl{3}  & \qw  & \qw  & \qw     & \ctrl{3} & \qw   & \qw   & \qw     & \qw   & \qw   & \qw   & \qw     & \qw   & \qw \\
    & \qw   & \qw  & \qw  & \qw     & \qw   & \qw   & \qw  & \ctrl{2} & \qw     & \qw  & \qw   & \ctrl{2} & \qw     & \qw   & \qw   & \qw   & \qw     & \qw   & \qw \\
    & \qw   & \qw  & \ctrl{1} & \qw     & \qw   & \qw   & \qw  & \qw  & \qw     & \qw  & \qw   & \qw   & \qw     & \qw   & \qw   & \ctrl{1}  & \qw     & \qw   & \qw \\
    & \gate{H}  & \gate{T}  & \targ{}  & \gate{T^\dagger} & \gate{H}  & \targ{} & \gate{T}  & \targ{} & \gate{T^\dagger} & \targ{} & \gate{T}  & \targ{} & \gate{T^\dagger}  & \gate{H}  & \gate{T}  & \targ{} & \gate{T^\dagger}  & \gate{H}  & \qw 
    \end{quantikz}
\end{adjustbox}
\\ 
\vspace{.5cm}
\begin{adjustbox}{width=.42\textwidth}
(c)
    \begin{quantikz}[column sep=10pt, row sep={20pt,between origins}]
    & \ctrl{1} &  \\
    & \ctrl{1} & \\
    & \ctrl{1} & \\
    & \targ{} & 
    \end{quantikz}
    =
    \begin{quantikz}[column sep=10pt, row sep={20pt,between origins}]
    & \gate[3]{RTOF3} & \qw & \gate[3]{RTOF3^\dag} & \qw \\
    & & \qw & \ctrl{1} & \qw \\
    \lstick{$\ket{0}$} & & \ctrl{1} & & \qw \\
    &  & \ctrl{1} & \qw & \qw \\
    & \qw & \targ{} & \qw & \qw 
    \end{quantikz}
\end{adjustbox}
\hspace{2cm}
\begin{adjustbox}{width=.42\textwidth}
(d)
    \begin{quantikz}[column sep=10pt, row sep={20pt,between origins}]
    & \ctrl{1} & \\
    & \ctrl{1} & \\
    & \ctrl{1} & \\
    & \ctrl{1} & \\
    & \targ{} & 
    \end{quantikz}
    =
    \begin{quantikz}[column sep=10pt, row sep={20pt,between origins}]
    & \gate[4]{RTOF4} & \qw & \gate[4]{RTOF4^\dag} & \qw \\
    &  & \qw & & \qw \\
    &  & \qw & & \qw \\
    \lstick{$\ket{0}$} & & \ctrl{1} & & \qw \\
    & \qw & \ctrl{1} & \qw & \qw \\
    & \qw & \targ{} & \qw & \qw
    \end{quantikz}
\end{adjustbox}
\caption{(a) Three-qubit Toffoli gate implemented up to a relative phase and with optimal CX count. (b) Four-qubit Toffoli gate implemented up to a relative phase. (c) Four-qubit Toffoli and (d) five-qubit Toffoli implemented using a single ancilla qubit. RTOF3 is self-adjoint, but we keep the $\dagger$ symbol as a guidance for the recursive construction of high-order Toffoli gates explained in the text. All circuits are taken from Ref.~\cite{Maslov_2016}.}
\label{fig:rtof}
\end{figure}

\paragraph{Compilation of the swapper circuit.} The swapper circuit is used in all our experiments. For subsets of size $2^{n-1}$ the optimal swapper is the $n$-qubit Grover diffusion operator, up to a phase. In the main text, Fig.~1D, right panel, shows the high level circuit for $n=4$. Clearly the main implementation bottleneck is the multi-controlled Toffoli gate. We compiled the circuit to the native gate set of Quantinuum H-series for several values of $n$ using \texttt{pytket}. We found that the number of two-qubit gates almost exactly matches the curve $4n^2 - 4n - 0.5$. However, in our game the player is allowed to use ancilla qubits as a workspace for a more efficient compilation. While there exist a plethora of techniques that could be used to compile multi-controlled Toffoli gates, we choose to implement the method in Ref.~\cite{Maslov_2016}. In particular we utilize the following result.

\begin{theorem}[restatement of Prop.~4 in Ref.~\cite{Maslov_2016}]
A size $n\geq 4$ multi-controlled Toffoli gate can be implemented by a circuit that uses $\lceil \frac{n-3}{2} \rceil$ ancillary qubits initialized in and returned to $\ket{0}$, $8n- 17$ T gates, $6n - 12$ CX gates, and $4n - 10$ Hadamard gates.
\end{theorem}

The author gives a constructive proof which uses as building blocks the three-qubit Toffoli (TOF3), the three-qubit Toffoli up to a relative phase (RTOF3) and the four-qubit Toffoli up to a relative phase (RTOF4) illustrated in Supplementary Fig.~\ref{fig:rtof}(a) and (b). In Supplementary Fig.~\ref{fig:rtof}(c) we show the implementation of the four-qubit Toffoli (TOF4) using a single ancilla qubit. The top two qubits are used as controls for the RTOF3 with the ancilla as target. We then apply a TOF3 and finally we undo the RTOF3. To go from TOF4 to TOF5, Supplementary Fig.~\ref{fig:rtof}(d), we simply exchange the RTOF3s to RTOF4s. For a six-qubit Toffoli (not shown) we add an ancilla and exchange the TOF3 in the middle with a TOF4. As we increase the number of qubits the recursion rule is the following: i) when going from an odd number of controls to  even, exchange the outermost RTOF3s to RTOF4s, and ii) when going from an even number of controls to odd, add an ancilla qubit and exchange the TOF3 in the middle to a TOF4.

\paragraph{Compilation of DES2 operations.} 
DES2 operations are used in our experiments with PRPs. DES2 operations map any 3-bits computational basis state $\ket{b_1, b_2, b_3}$ to $\ket{b_1, b_2, b_3 \oplus f(b_1, b_2)}$ up to arbitrary order of the output bits. The reordering is implemented by at most two SWAP gates. On the H2 devices, the SWAP gate is handled by relabeling and transporting qubits, therefore not incurring any extra cost on the two-qubit gate count~\cite{Moses_2023}. The function is defined as $f:\Set{0,1}^2\to \Set{0,1}$, so there are $2^{2^2}= 16$ of them. Supplementary Table~\ref{tab:des2_functions_dec} shows all possible functions along with a decomposition in terms of X, CX and CCX gates. It is known that a CCX, also known as Toffoli gate, can be decomposed into five generic two-qubit gates~\cite{DiVincenzo_1998} and that this number is necessary~\cite{Yu_2013}. For completeness we report a possible two-qubits gate decomposition using the native gate set of Quantinuum H-series hardware in Supplementary Fig.~\ref{fig:h2_gateset_dec}. We conclude that DES2 operations require five two-qubit gates in the worst case, and three in average.

\begin{table}
\centering
\small
\begin{tabular}{l|l|l}
$f(b_1, b_2)$ & decomposition & two-qubit gates\\
\hline
0 & I & 0\\
$b_1 \land b_2$ & CCX($b_1, b_2, b_3$) & 5\\
$b_1 \land \lnot b_2$ & X($b_2$)CCX($b_1, b_2, b_3$)X($b_2$) & 5\\
$b_1$ & CX($b_1, b_3$) & 1\\
$\lnot b_1 \land  b_2$ & X($b_1$)CCX($b_1, b_2, b_3$)X($b_1$) & 5\\
$b_2$ & CX($b_2, b_3$) & 1\\
$b_1 \oplus b_2$ & CX($b_1, b_3$)CX($b_2, b_3$) & 2\\
$b_1 \lor b_2$ & X($b_1$)X($b_2$)CCX($b_1, b_2, b_3$)X($b_3$)X($b_2$)X($b_1$)& 5\\
$\lnot(b_1 \lor b_2)$ & X($b_1$)X($b_2$)CCX($b_1, b_2, b_3$)X($b_2$)X($b_1$)& 5\\
$\lnot(b_1 \oplus b_2)$ & CX($b_1, b_3$)CX($b_2, b_3$)X($b_3$) & 2\\
$\lnot b_2$ & X($b_2$)CX($b_2, b_3$)X($b_2$) & 1\\
$b_1 \lor \lnot b_2$ & X($b_1$)CCX($b_1, b_2, b_3$)X($b_3$)X($b_1$)& 5\\
$\lnot b_1$ & X($b_1$)CX($b_1, b_3$)X($b_1$) & 1\\
$\lnot b_1 \lor b_2$ & X($b_2$)CCX($b_1, b_2, b_3$)X($b_3$)X($b_2$)& 5\\
$\lnot(b_1 \land b_2)$ & CCX($b_1, b_2, b_3$)X($b_3$) & 5\\
1 & X($b_3$) & 0\\
\end{tabular}
\caption{Decomposition of DES2 gates to quantum circuits in the gate set X, CX, and CCX.}
\label{tab:des2_functions_dec}
\end{table}

\begin{figure}
\centering
\begin{adjustbox}{width=.65\textwidth}
    \begin{quantikz}[column sep=10pt, row sep={20pt,between origins}]
    & \ctrl{1} &  \\
    & \targ{} &
    \end{quantikz}
    =
    \begin{quantikz}[column sep=0.25cm, row sep={1cm, between origins}]
    &  & \gate[2]{R_{ZZ}(\frac{\pi}{2})} & & \gate{R_Z(\frac{\pi}{2})} &\\
    & \gate{U(\frac{5\pi}{2}, \frac{\pi}{2})} & & \gate{U(\frac{\pi}{2}, 0)} & \gate{R_Z(\frac{3\pi}{2})} &
    \end{quantikz}
\end{adjustbox}
\\ 
\vspace{.65cm}
\begin{adjustbox}{width=1\textwidth}
    \begin{quantikz}[column sep=0.2cm, row sep={1cm, between origins}]
    & \ctrl{1} &  \\
    & \ctrl{1} & \\
    & \targ{} & 
    \end{quantikz}
    =
    \begin{quantikz}[column sep=0.2cm, row sep={1cm, between origins}, transparent]
    & & & & \gate[3, label style={yshift=0.3cm}]{R_{ZZ}(\frac{\pi}{2})} & & & & \gate[3, label style={yshift=0.3cm}]{R_{ZZ}(\frac{\pi}{2})} & \gate{U(\pi, \frac{\pi}{4})} & \gate[2]{R_{ZZ}(\frac{\pi}{4})} & \gate{U(\pi, 1.75\pi)} & \gate{R_Z(\frac{\pi}{4})} & \\
    & & \gate[2]{R_{ZZ}(\frac{\pi}{2})} & & \linethrough & &  \gate[2]{R_{ZZ}(\frac{\pi}{2})} & & \linethrough & & & & \gate{R_Z(3.25\pi)} &\\
    & \gate{U(\pi, \frac{3\pi}{2})} & & \gate{U(\frac{\pi}{4}, \frac{\pi}{2})} & & \gate{U(\frac{\pi}{4}, 0)} & & \gate{U(\frac{\pi}{4}, \frac{3\pi}{2})} & & \gate{U(3.25\pi,\pi)} & & & \gate{R_Z(\pi)} &
    \end{quantikz}
\end{adjustbox}
\caption{Decomposition of CX (top) and CCX (bottom) gates using the native gate set of the Quantinuum H-series hardware. Here $R_{ZZ}(\theta) = e^{-i \frac{\theta}{2} Z\otimes Z}$, $R_{Z}(\theta) = e^{-i \frac{\theta}{2} Z}$ and $U(\theta, \phi) = e^{-i\frac{\theta}{2} (\cos(\phi)X + \sin(\phi)Y)}$.}
\label{fig:h2_gateset_dec}
\end{figure}
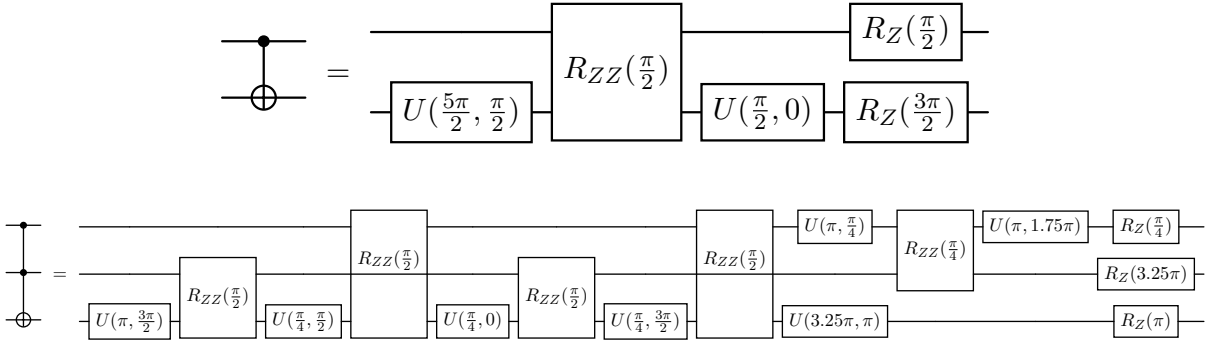

\paragraph{Compilation of S8 operations.}
S8 operations are only used in our experiments with PRPs. S8 operations permute the computational basis states of 3 bits. We approach the problem of compiling these operations numerically by brute-force. We use the \texttt{pytket} software package~\cite{Sivarajah_2021} and compile for the native Quantinuum H-series gate set. For each of the $8! = 40320$ permutations we construct the corresponding unitary using the \texttt{ToffoliBox} tool in \texttt{pytket}. We then compile each unitary to a quantum circuit and optimize it with the highest level of optimization (three), which consists of a collection of gate reduction techniques. Supplementary Fig.~\ref{fig:S8_compiled} shows the histogram of how many permutations yielded a certain number of two-qubit $R_{ZZ}(\theta)$ gates. We find that S8 operations require at most 18 two-qubit gates, and 10 on average.

\begin{figure}
    \centering    
    \includegraphics[width=.5\linewidth]{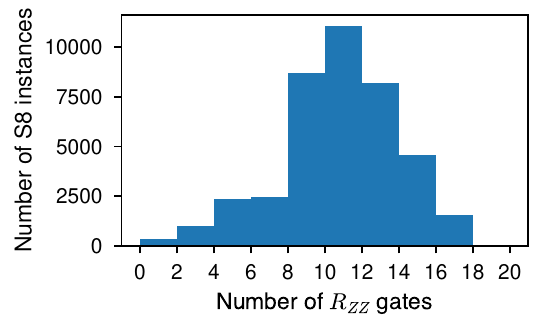}
    \vspace{-.25cm}
    \caption{Histogram of S8 instances that require a certain number of $R_{ZZ}(\theta)$ two-qubit gates.}
    \label{fig:S8_compiled}
\end{figure}

\paragraph{Code's functionality.} In~\cite{zenodo_complement_sampling2025} we provide access to the package that implements the circuits described in the previous sections. The user can run the example notebooks to generate the full circuits for the desired experiment and then inspect the results. For the hardware experiments, the authors ran the circuits on Quantinuum System Model H2 quantum computers via Quantinuum Nexus~\cite{quantinuum_nexus}, but other backends could be selected to test such experiments.


\begin{thebibliography}{10}

\bibitem{Bell_1964}
J.~S. Bell, ``{On the Einstein Podolsky Rosen paradox},'' \href{https://dx.doi.org/10.1103/PhysicsPhysiqueFizika.1.195}{{\em Physics Physique Fizika} {\bfseries 1}, 195--200 (1964)}.

\bibitem{Freedman_1972}
S.~J. Freedman and J.~F. Clauser, ``{Experimental Test of Local Hidden-Variable Theories},'' \href{https://dx.doi.org/10.1103/PhysRevLett.28.938}{{\em Phys. Rev. Lett.} {\bfseries 28}, 938--941 (1972)}.

\bibitem{Aspect_1982}
A.~Aspect, P.~Grangier, and G.~Roger, ``{Experimental Realization of Einstein-Podolsky-Rosen-Bohm Gedankenexperiment: A New Violation of Bell's Inequalities},'' \href{https://dx.doi.org/10.1103/PhysRevLett.49.91}{{\em Phys. Rev. Lett.} {\bfseries 49}, 91--94 (1982)}.

\bibitem{Pan_1998}
J.-W. Pan, D.~Bouwmeester, H.~Weinfurter, and A.~Zeilinger, ``{Experimental Entanglement Swapping: Entangling Photons That Never Interacted},'' \href{https://dx.doi.org/10.1103/PhysRevLett.80.3891}{{\em Phys. Rev. Lett.} {\bfseries 80}, 3891--3894 (1998)}.

\bibitem{Storz_2023}
S.~Storz, J.~Sch{\"a}r, {\em et~al.}, ``{Loophole-free Bell inequality violation with superconducting circuits},'' \href{https://dx.doi.org/10.1038/s41586-023-05885-0}{{\em Nature} {\bfseries 617}, no.~7960, 265--270 (2023)}.

\bibitem{Brunner_2014}
N.~Brunner, D.~Cavalcanti, S.~Pironio, V.~Scarani, and S.~Wehner, ``Bell nonlocality,'' \href{https://dx.doi.org/10.1103/revmodphys.86.419}{{\em Reviews of Modern Physics} {\bfseries 86}, no.~2, 419--478 (2014)}.

\bibitem{Arute_2019}
F.~Arute, K.~Arya, {\em et~al.}, ``Quantum supremacy using a programmable superconducting processor,'' \href{https://dx.doi.org/10.1038/s41586-019-1666-5}{{\em Nature} {\bfseries 574}, 505--510 (2019)}.

\bibitem{Zhong_2020}
H.-S. Zhong, H.~Wang, {\em et~al.}, ``Quantum computational advantage using photons,'' \href{https://dx.doi.org/10.1126/science.abe8770}{{\em Science} {\bfseries 370}, no.~6523, 1460--1463 (2020)}.

\bibitem{Madsen_2022}
L.~S. Madsen, F.~Laudenbach, {\em et~al.}, ``Quantum computational advantage with a programmable photonic processor,'' \href{https://dx.doi.org/10.1038/s41586-022-04725-x}{{\em Nature} {\bfseries 606}, no.~7912, 75--81 (2022)}.

\bibitem{DeCross_2025}
M.~DeCross, R.~Haghshenas, {\em et~al.}, ``Computational power of random quantum circuits in arbitrary geometries,'' \href{https://dx.doi.org/10.1103/PhysRevX.15.021052}{{\em Phys. Rev. X} {\bfseries 15}, 021052 (2025)}.

\bibitem{Kretschmer_2025}
W.~Kretschmer, S.~Grewal, {\em et~al.}, ``Demonstrating an unconditional separation between quantum and classical information resources,'' \href{https://arxiv.org/abs/2509.07255}{{\ttfamily arXiv:2509.07255 [quant-ph]}} (2025).

\bibitem{Harrow_2017}
A.~W. Harrow and A.~Montanaro, ``Quantum computational supremacy,'' \href{https://dx.doi.org/10.1038/nature23458}{{\em Nature} {\bfseries 549}, no.~7671, 203--209 (2017)}.

\bibitem{Zlokapa_2023}
A.~Zlokapa, B.~Villalonga, S.~Boixo, and D.~A. Lidar, ``Boundaries of quantum supremacy via random circuit sampling,'' \href{https://dx.doi.org/10.1038/s41534-023-00703-x}{{\em npj Quantum Information} {\bfseries 9}, no.~1, 36 (2023)}.

\bibitem{StilckFranca_2022}
D.~Stilck~Fran{\c{c}}a and R.~Garcia-Patron, ``A game of quantum advantage: linking verification and simulation,'' \href{https://dx.doi.org/10.22331/q-2022-06-30-753}{{\em {Quantum}} {\bfseries 6}, 753 (2022)}.

\bibitem{Hangleiter_2019}
D.~Hangleiter, M.~Kliesch, J.~Eisert, and C.~Gogolin, ``Sample complexity of device-independently certified ``quantum supremacy'','' \href{https://dx.doi.org/10.1103/PhysRevLett.122.210502}{{\em Phys. Rev. Lett.} {\bfseries 122}, 210502 (2019)}.

\bibitem{Aaronson_2017}
S.~Aaronson and L.~Chen, \href{https://dx.doi.org/10.4230/LIPIcs.CCC.2017.22}{``{Complexity-Theoretic Foundations of Quantum Supremacy Experiments},''} in {\em 32nd Computational Complexity Conference (CCC 2017)}, vol.~79 of {\em Leibniz International Proceedings in Informatics (LIPIcs)}, pp.~22:1--22:67.
\newblock Schloss Dagstuhl -- Leibniz-Zentrum f{\"u}r Informatik, Dagstuhl, Germany, 2017.

\bibitem{benedetti2025complement}
M.~Benedetti, H.~Buhrman, and J.~Weggemans, ``Provable and verifiable quantum advantage in sample complexity,'' \href{https://dx.doi.org/10.1103/q55v-wm7y}{{\em Phys. Rev. Lett.} {\bfseries 136}, 040601 (Jan, 2026)}.

\bibitem{aaronson2020hardness}
S.~Aaronson, Y.~Atia, and L.~Susskind, ``On the hardness of detecting macroscopic superpositions,'' \href{https://arxiv.org/abs/2009.07450}{{\ttfamily arXiv:2009.07450 [quant-ph]}} (2020).

\bibitem{Buhrman_2012}
H.~Buhrman, O.~Regev, G.~Scarpa, and R.~d. Wolf, ``{Near-Optimal and Explicit Bell Inequality Violations},'' \href{https://dx.doi.org/10.4086/toc.2012.v008a027}{{\em Theory of Computing} {\bfseries 8}, no.~27, 623--645 (2012)}.

\bibitem{katz2020introduction}
J.~Katz and Y.~Lindell, {\em Introduction to Modern Cryptography}.
\newblock Chapman \& Hall/CRC Cryptography and Network Security Series. CRC Press, 2020.

\bibitem{Bernstein_1997}
E.~Bernstein and U.~Vazirani, \href{https://dx.doi.org/10.1145/167088.167097}{``Quantum complexity theory,''} in {\em Proceedings of the Twenty-Fifth Annual ACM Symposium on Theory of Computing}, STOC '93, p.~11–20.
\newblock Association for Computing Machinery, New York, NY, USA, 1993.

\bibitem{Moses_2023}
S.~A. Moses, C.~H. Baldwin, {\em et~al.}, ``A race-track trapped-ion quantum processor,'' \href{https://dx.doi.org/10.1103/PhysRevX.13.041052}{{\em Phys. Rev. X} {\bfseries 13}, 041052 (2023)}.

\bibitem{Magesan_2011}
E.~Magesan, J.~M. Gambetta, and J.~Emerson, ``Scalable and robust randomized benchmarking of quantum processes,'' \href{https://dx.doi.org/10.1103/PhysRevLett.106.180504}{{\em Phys. Rev. Lett.} {\bfseries 106}, 180504 (2011)}.

\bibitem{Harper_2019}
R.~Harper, I.~Hincks, C.~Ferrie, S.~T. Flammia, and J.~J. Wallman, ``Statistical analysis of randomized benchmarking,'' \href{https://dx.doi.org/10.1103/PhysRevA.99.052350}{{\em Phys. Rev. A} {\bfseries 99}, 052350 (2019)}.

\bibitem{hseries_benchmarking}
Quantinuum, 2025.
\newblock \url{https://docs.quantinuum.com/systems/user_guide/hardware_user_guide/performance_validation.html} [Accessed: 2025-10-16].

\bibitem{Cai_2023}
Z.~Cai, R.~Babbush, S.~C. Benjamin, S.~Endo, W.~J. Huggins, Y.~Li, J.~R. McClean, and T.~E. O'Brien, ``Quantum error mitigation,'' \href{https://dx.doi.org/10.1103/RevModPhys.95.045005}{{\em Rev. Mod. Phys.} {\bfseries 95}, 045005 (2023)}.

\bibitem{Bennett_1993}
C.~H. Bennett, G.~Brassard, C.~Cr\'epeau, R.~Jozsa, A.~Peres, and W.~K. Wootters, ``{Teleporting an unknown quantum state via dual classical and Einstein-Podolsky-Rosen channels},'' \href{https://dx.doi.org/10.1103/PhysRevLett.70.1895}{{\em Phys. Rev. Lett.} {\bfseries 70}, 1895--1899 (1993)}.

\bibitem{Maslov_2016}
D.~Maslov, ``{Advantages of using relative-phase Toffoli gates with an application to multiple control Toffoli optimization},'' \href{https://dx.doi.org/10.1103/physreva.93.022311}{{\em Physical Review A} {\bfseries 93}, no.~2, 022311 (2016)}.

\bibitem{Foreman_2023}
C.~Foreman, S.~Wright, A.~Edgington, M.~Berta, and F.~J. Curchod, ``Practical randomness amplification and privatisation with implementations on quantum computers,'' \href{https://dx.doi.org/10.22331/q-2023-03-30-969}{{\em {Quantum}} {\bfseries 7}, 969 (2023)}.

\bibitem{efronIntroductionBootstrap1994}
B.~Efron and R.~J. Tibshirani, \href{https://dx.doi.org/10.1201/9780429246593}{{\em An {{Introduction}} to the {{Bootstrap}}}}.
\newblock {Chapman and Hall/CRC}, New York, 1994.

\bibitem{hseries_quantumvolume}
Quantinuum, ``H-series system benchmarks.'' 2025.
\newblock \url{https://docs.quantinuum.com/systems/user_guide/hardware_user_guide/benchmarks/system_benchmarks.html} and \url{https://github.com/CQCL/quantinuum-hardware-quantum-volume} [Accessed: 2025-09-18].

\bibitem{gay2025pseudorandomness}
W.~Gay, W.~He, N.~Kocurek, and R.~O'Donnell, \href{https://dx.doi.org/10.1007/978-3-032-01855-7_21}{``Pseudorandomness properties of random reversible circuits,''} in {\em Advances in Cryptology -- CRYPTO 2025}, pp.~651--678.
\newblock Springer Nature Switzerland, Cham, 2025.

\bibitem{Bombin_2006}
H.~Bombin and M.~A. Martin-Delgado, ``Topological quantum distillation,'' \href{https://dx.doi.org/10.1103/PhysRevLett.97.180501}{{\em Phys. Rev. Lett.} {\bfseries 97}, 180501 (2006)}.

\bibitem{Bravyi_2005}
S.~Bravyi and A.~Kitaev, ``Universal quantum computation with ideal clifford gates and noisy ancillas,'' \href{https://dx.doi.org/10.1103/PhysRevA.71.022316}{{\em Phys. Rev. A} {\bfseries 71}, 022316 (2005)}.

\bibitem{gosset2025multiqubit}
D.~Gosset, R.~Kothari, and C.~Zhang, ``{Multi-qubit Toffoli with exponentially fewer T gates},'' \href{https://arxiv.org/abs/2510.07223}{{\ttfamily arXiv:2510.07223 [quant-ph]}} (2025).

\bibitem{Hoyer_2005}
P.~H{\o}yer and R.~{\v S}palek, ``Quantum fan-out is powerful,'' \href{https://dx.doi.org/10.4086/toc.2005.v001a005}{{\em Theory of Computing} {\bfseries 1}, no.~5, 81--103 (2005)}.

\bibitem{zenodo_complement_sampling2025}
M.~Rosenkranz, G.~Marin-Sanchez, and M.~Benedetti, ``{Supporting material for `Unconditional and exponentially large violation of classicality'}.'' 2026.
\newblock Zenodo. \url{https://doi.org/10.5281/zenodo.17593255}.

\bibitem{luby1988how}
M.~Luby and C.~Rackoff, ``How to construct pseudorandom permutations from pseudorandom functions,'' \href{https://dx.doi.org/10.1137/0217022}{{\em SIAM Journal on Computing} {\bfseries 17}, no.~2, 373--386 (1988)}.

\bibitem{Even_1983}
S.~Even and O.~Goldreich, ``{DES-like functions can generate the alternating group},'' \href{https://dx.doi.org/10.1109/TIT.1983.1056752}{{\em IEEE Transactions on Information Theory} {\bfseries 29}, no.~6, 863--865 (1983)}.

\bibitem{Sivarajah_2021}
S.~Sivarajah, S.~Dilkes, A.~Cowtan, W.~Simmons, A.~Edgington, and R.~Duncan, ``{t$|$ket⟩: a retargetable compiler for NISQ devices},'' \href{https://dx.doi.org/10.1088/2058-9565/ab8e92}{{\em Quantum Science and Technology} {\bfseries 6}, no.~1, 014003 (2020)}.

\bibitem{diakonikolas2017optimal}
I.~Diakonikolas, T.~Gouleakis, J.~Peebles, and E.~Price, ``Optimal identity testing with high probability,'' \href{https://arxiv.org/abs/1708.02728}{{\ttfamily arXiv:1708.02728}} (2017).

\bibitem{DiVincenzo_1998}
D.~P. DiVincenzo, ``Quantum gates and circuits,'' \href{https://dx.doi.org/10.1098/rspa.1998.0159}{{\em Proceedings of the Royal Society of London. Series A: Mathematical, Physical and Engineering Sciences} {\bfseries 454}, no.~1969, 261–276 (1998)}.

\bibitem{Yu_2013}
N.~Yu, R.~Duan, and M.~Ying, ``Five two-qubit gates are necessary for implementing the {Toffoli} gate,'' \href{https://dx.doi.org/10.1103/PhysRevA.88.010304}{{\em Phys. Rev. A} {\bfseries 88}, 010304 (2013)}.

\bibitem{quantinuum_nexus}
``{Quantinuum Nexus}.'' 2024.
\newblock \url{https://nexus.quantinuum.com/}.

\end{thebibliography}
\end{document}